\PassOptionsToPackage{twoside=false}{geometry}
\documentclass[nonacm, screen]{acmart}

\let\tilde\widetilde
\usepackage{amsthm}
\usepackage{bm}
\usepackage{amsmath}
\usepackage{mathtools}  
\usepackage{cleveref}
\usepackage{listings}
\usepackage{soul}

\usepackage{nicefrac}
\usepackage{mathtools}  
\usepackage{amsfonts}  
\usepackage{dsfont}
\usepackage{xspace}
\usepackage{subfiles}
\usepackage{mdframed}
\usepackage[skins,most]{tcolorbox}
\usepackage{framed}

\usepackage{tikz}
\usepackage{pgfplots}
\pgfplotsset{width=9cm,compat=1.8}
\usetikzlibrary{matrix,positioning, arrows}
\usetikzlibrary{decorations.pathreplacing}
\usetikzlibrary{decorations.pathmorphing}
\tikzset{snake it/.style={decorate, decoration=snake, thick}}
\usetikzlibrary{arrows,backgrounds,calc,trees}
\pgfdeclarelayer{background}
\pgfsetlayers{background,main}

\usepackage{booktabs} 
\usepackage{tabulary}
\usepackage{tabularx}
\usepackage{threeparttable}
\usepackage{etoolbox}
\usepackage{enumitem}

\usepackage{mdframed}
\mdfdefinestyle{theoremstyle}{
    roundcorner=5pt,
    linecolor=black,
    linewidth=1pt,
    backgroundcolor=gray!10,
    innertopmargin=0pt,
    innerbottommargin=6pt,
    leftmargin=1pt,
    rightmargin=1pt
}

\usepackage{float}
\floatname{algoflt}{Protocol}
\floatstyle{ruled}
\newfloat{algoflt}{tbp}{loa}
\makeatletter
\lst@AddToHook{OnEmptyLine}{\addtocounter{lstnumber}{-1}\vspace{-0.5\baselineskip}}

\lstnewenvironment{algorithm}[2][H]{\lstset{mathescape=true,tabsize=2,numbers=none,numberstyle=\tiny\dcmnumbersty,numberblanklines=false,showlines=false,basicstyle=\ttfamily\footnotesize,backgroundcolor={},columns=fullflexible,emph={[1]if, else, then, and, or, for, do, while, return, exit, not, output, initialize, var },
    emphstyle={[1]\bfseries},escapeinside={/*}{*/},
    escapebegin={\ifmmode\else\hfill\color{gray}\small$\triangleright$ \fi},
    escapeend={\ifmmode\else\fi},
    }\setlength{\topsep}{0pt}\setlength{\parindent}{0pt}\@skiphyperreftrue\algoflt[#1]\caption{#2}}{\endalgoflt\@skiphyperreffalse}
\makeatother

\newcommand{\Autoref}[1]{\begingroup
  \def\chapterautorefname{Chapter}\def\sectionautorefname{Section}\def\subsectionautorefname{Subsection}\def\figureautorefname{Figure}\def\tableautorefname{Table}\autoref{#1}\endgroup
}

\newcommand{\vx}{\Vec{x}}
\newcommand{\vX}[1]{{\X{#1}}}

\DeclarePairedDelimiter\abs{\lvert}{\rvert}

\DeclarePairedDelimiterX{\inp}[2]{\langle}{\rangle}{#1, #2}

\newcommand{\xmax}[1]{\ensuremath{{x}_{\mathsf{max}}}}
\newcommand{\xmin}[1]{\ensuremath{{x}_{\mathsf{min}}}}

\newcommand{\E}[1]{\ensuremath{{\textbf{E}\left[#1\right]}}}

\newcommand{\pr}[1]{\ensuremath{\textbf{Pr}\left[#1\right]}}

\newcommand{\x}[1]{\ensuremath{\vec{{x}}({#1})}}

\newcommand{\X}[1]{\ensuremath{\vec{{X}}\left({#1}\right)}}
\newcommand{\disc}[1]{\ensuremath{\mathbf{disc}}\left(#1\right)}

\newcommand{\vol}[1]{\ensuremath{{\rm  vol}\!\left(#1\right)}} \newcommand{\voll}{\ensuremath{{\rm  vol}}}

\newcommand{\dmin}{\ensuremath{{d}_{\mathsf{min}}(G)}}
\newcommand{\dave}{\ensuremath{{d}_{\mathsf{ave}}(G)}}
\newcommand{\dminNG}{\ensuremath{{d}_{\mathsf{min}}}}
\newcommand{\dmaxNG}{\ensuremath{{d}_{\mathsf{max}}}}
\newcommand{\daveNG}{\ensuremath{{d}_{\mathsf{ave}}}}

\renewcommand{\max}[1]{\ensuremath{\underset{#1}{\text{max}}}\xspace}

\DeclareMathOperator{\sign}{sign}

\usepackage{xcolor}
\usepackage{color}

\newcommand{\ovl}[1]{\overline{#1}}

\newcommand{\protocol}[1]{\textsc{#1}\xspace}
\newcommand{\whp}{\text{w.h.p.}\xspace}
\usepackage{thm-restate}
\usepackage{lipsum}
\usepackage{wrapfig}

\declaretheorem[name=Claim]{claim}
\declaretheorem[name=Observation]{observation}

\newcommand{\ProtAV}{\textsc{Asynchro\-nous-DIV}\xspace}

\newcommand{\brac}[1]{\left( #1 \right)}

\allowdisplaybreaks
\author{Petra Berenbrink}
\orcid{0000-0002-6930-3259}
\affiliation{\institution{University of Hamburg}
  \city{ Hamburg}
  \country{Germany}
}
\email{petra.berenbrink@uni-hamburg.de}

\author{Colin Cooper}
\orcid{0000-0002-5264-4401}
\affiliation{\institution{King's College London}
  \city{London}
  \country{United Kingdom}
}
\email{colin.cooper@kcl.ac.uk}

\author{Thorsten G\"otte}
\orcid{0000-0001-9798-6993}
\affiliation{\institution{University of Hamburg}
  \city{ Hamburg}
  \country{Germany}
}
\email{thorsten.goette@uni-hamburg.de }

\author{Lukas Hintze}
\orcid{0009-0006-8348-4638}
\affiliation{\institution{University of Hamburg}
  \city{ Hamburg}
  \country{Germany}
}
\email{lukas.rasmus.hintze@uni-hamburg.de}

\author{Tomasz Radzik}
\orcid{0000-0002-7776-5461}
\affiliation{\institution{King's College London}
  \city{London}
  \country{United Kingdom}
}
\email{tomasz.radzik@kcl.ac.uk}

                 \setcopyright{none}                             \renewcommand\footnotetextcopyrightpermission{} 

\makeatother

\title{Discrete Incremental Voting}
\subtitle{New Bounds for General Graphs and Expanders}

\thanks{Revised: \today}                                   

\begin{document}

\begin{abstract}

We analyze a process on a connected undirected graph
called discrete incremental voting. 
Each node 
holds an integer opinion. 
At each step of the process, a randomly selected node modifies its opinion 
by changing it  by one toward the opinion of a randomly chosen  neighbor if these opinions are different.
The process converges to a unique opinion that, in expectation, is the degree-weighted average of the initial opinions.

We show that for an $n$-node graph $G$ with conductance $\Phi$
and the maximum difference between initial opinions of~$K$,
the expected convergence time is
${O}\left({n\left(K\log (Kn)+\gamma n  \right)}/{\Phi^2}\right)$, where $\gamma$ is the ratio of the average to smallest degree of $G$.
For regular graphs with second-largest absolute eigenvalue  $o(1/\log^4 n)$, 
we  show two point concentration of the final opinion at the integer average provided $K=o(n/\log^2 n)$.
\end{abstract}

\maketitle

\section{Introduction}

A voting process operates on an undirected connected graph $G = (V,E)$ with $n$ nodes
and $m$ edges, where the nodes represent agents,
each holding its opinion. These agents interact with their {neighbors} via edges to reach  agreement on a single opinion.
Such processes arise in a wide range of settings, including distributed and parallel computing, as well as
the social sciences.
In parallel and distributed systems, voting serves as a fundamental primitive for various tasks, e.g.,
for the leader-election task where the system must agree on a unique coordinator.
At the same time, voting processes provide abstract models of real-world opinion dynamics,
in which agents repeatedly exchange information and seek to converge to a common
decision.

The pull voter (or pull voting) model is 
perhaps the most well-known voting process in parallel and
distributed computing, see \cite{DBLP:journals/iandc/HassinP01}.
In each discrete time step,
a randomly chosen node $u\in V$ selects a random neighbor $v \in N(u) = \{x\in V: \{u,x\}\in E\}$ 
and adopts its opinion.
The process converges
to a consensus state in which all nodes hold the same opinion, and 
the convergence time can be characterized via the coalescence time of 
independent random walks~\cite{aldous-fill-unpublished,DBLP:journals/siamdm/CooperEOR13}.
Moreover, the probability that a given opinion is adopted is proportional to the initial
degree of the nodes holding that opinion~\cite{DBLP:journals/iandc/HassinP01}.

In computer science, the opinions in pull voting do not necessarily have intrinsic meaning or values 
associated with them. Adopting a different opinion is just
a device for moving a distributed system towards consensus.
In the social sciences however, models of opinion dynamics typically incorporate some notion of
compromise when agents interact, rather than simple adoption of opinions held by others. In this context consensus means convergence to an agreed opinion, for example 
by changing incrementally or by averaging initial opinions, rather than the elimination of competitors as a result of random interactions.

Computing an average in a distributed system can also be regarded as a consensus process, the final average being the consensus value.
Typically averaging is done by pairwise arithmetic interactions across edges.
{What is not immediately obvious}, but is a topic of this paper, is that with some possible loss of accuracy,  such a task can be  easily  carried out by  an adapted  pull voting process. 

In a statistical or data analysis setting, opinion data is typically divided into three groups, namely:
categorical (incomparable, in the sense that all we can know about two different opinions is that they are not the same),
ordinal (can be ordered) and interval (numeric). The boundaries between these groups are not precise. 
Categorical data such as colors (red, blue, green) can be ordered in two dimensions. 
A five-point Likert scale (strongly disagree, disagree, undecided, agree, strongly agree) falls between ordinal and interval.

In this paper we consider a voting process which operates on  data consisting of integers in a given range.
 This process,  \emph{discrete incremental voting} (DIV), was 
introduced
in~\cite{CRS23}. 
It can incorporate Likert  and integer scales,
and can be viewed as bridging the gap between the viewpoints of opinion dynamics in computer science, social sciences, and averaging.

In DIV, nodes hold opinions from $[K]\coloneqq \{0,1,\ldots,K\}$.
As in pull voting,  at each step, a randomly chosen node $u$ selects  a random neighbor $v \in N(u)$. 
Instead of adopting the opinion of $v$ outright, node $u$ moves one increment in the
direction of that opinion. Thus $u$
updates its opinion by $+1$ if $v$'s opinion is larger,
or  by $-1$ if $v$'s opinion is smaller.
This process is
formalized below in Protocol~\ref{prot:edgeinteraction} as \ProtAV.

As shown in \cite{CRS23}  DIV eventually converges  to a state in which all nodes hold the same opinion.
The central questions are
to determine the time to reach consensus, the \emph{convergence time},  
and characterize  the final opinion held by the nodes.
Despite the simplicity of this process, the convergence time of DIV
is less well understood than that  of many pull voting models.
The existing bounds, derived in \cite{CRS23,CRS24}, are either crude 
or apply  to restricted  classes  of graphs,
and to limited initial range of opinions $K$.

\begin{algorithm}{The \ProtAV process\label{prot:edgeinteraction} in a graph $G(V,E)$}
$\rhd\;\;\; x_u(t)$ is opinion of $u$ at (the beginning of) step $t$
$\rhd\;$ step $t$, changing configuration $\vx(t)=(x_v(t))_{v\in V}$ to $\vx(t+1)$

for $t = 0, 1,2, \ldots$
    Pick a node $u \in V$ uniformly at random
    Node $u$ chooses a neighbor $v \in N(u)$ uniformly at random
    if $x_u(t) > x_v(t)$: $\;\;\;\;\;\;\;\;\: x_u(t+1) \gets x_u(t) - 1$   
    if $x_u(t) < x_v(t)$: $\;\;\;\;\;\;\;\;\: x_u(t+1) \gets x_u(t) + 1$
    if $x_u(t) = x_v(t)$: $\;\;\;\;\;\;\;\;\: x_u(t+1) \gets x_u(t)$

    for all  $w \in V\setminus\{u\}$:     	$x_w(t+1) \gets x_w(t)$
\end{algorithm}

The main attraction of the DIV process is that on graphs with reasonable expansion, the opinions held by the individual nodes converge (with high probability) to the initial integer average opinion. 
As any change in node opinion is based purely on observing a neighbor's value, without any attempt 
at compensating this with 'load-balancing' updates 
at the neighbors, one might wonder how this can be.
Viewed as a black box hiding the graph and showing only the changes to the sum of all node opinions, 
the process outputs a sequence of values from $\{-1,0,+1\}$ depending on random choices made in the above algorithm.
In the case of regular graphs 
the random choices in one step are equivalent to choosing a random end point of a random edge. The output is thus equally likely to be $+1$ or $-1$ whenever the opinions of the edge endpoints differ,
and the output of the black box is a random walk on the integers starting from the initial sum of all opinions. 
This overall behavior is indifferent to the fact that at a detailed level the opinions held by the individual nodes 
are changing in a way that depends entirely on the weights of their neighbors. 

As in pull voting, the range of opinions reduces over time, so that after some time $T$ a state of {\em near consensus} is achieved where all opinions lie in a small interval, for example, they all are $i$ or $i+1$. 
Provided  the time $T$ for the process to reach this near consensus is small enough, a
 drift of $O(\sqrt{T})$ in the total sum of opinions
 caused by the associated random walk on the integers
 equates to an insignificant $O(\sqrt T/n)$ deviation of the current average of opinions from the initial one, 
 forcing the initial average  to lie
 in the near consensus interval $[i,i+1]$. At this point the process reverts to pull voting competing between opinions $i$ and $i+1$ which bracket the initial average. As consensus at or near the initial  average is assumed to be a desirable outcome,  the shorter the time to near convergence, the better the accuracy of the final opinion of DIV process.

{In this paper we present an improved general bound on the convergence time of discrete incremental voting on connected graphs.}
For any such graph $G$,
we prove that the expected convergence time is
${O}\left(n({K \log (Kn)+n\gamma(G)})/{\Phi(G)^2}\right)$,
where $\Phi(G)$ is the conductance of $G$
and $\gamma(G)$ is the ratio of the average to smallest node degree.
This general bound should be compared with 
the general bound 
of $O\left(K\cdot\mathcal{T}^{2V}_G\right)= O\left({Kn^2\gamma(G)}/{\Phi(G)}\right)$
from~\cite{CRS23}.
The term $\mathcal{T}^{2V}_G$ is the expected convergence time for pull voting 
with  two  opinions (2-value pull voting), and its upper bound of
$O\left({Kn^2\gamma(G)}/{\Phi(G)}\right)$
was proven in~\cite{CooperR16}.
Roughly speaking, disregarding
logarithmic terms, our new bound can be viewed as replacing $K\cdot n\gamma(G)$
in the general bound~from~\cite{CRS23}
with $(K + n\gamma(G))/\Phi(G)$, 
improving the previous bound when $\Phi(G) \gg 1/(n\gamma(G))$ and
$K \gg 1/\Phi(G)$.

We also prove a lower bound of 
$\Omega(Kn/\Phi(G)+\mathcal{T}^{2V}_G)$ for a family of regular graphs and worst-case initial
assignments of opinions to nodes. 
For constant conductance, the upper and lower bounds 
agree up to a logarithmic factor; for smaller conductance, the  terms in $K$ differ by an additional factor $1/\Phi(G)$.

For expander graphs
with $\lambda(G) = o(1/\log^4 n)$
and $K=o(n/\log^2 n)$,
we show that with high probability the final consensus opinion is either $\lfloor W(0)\rfloor$ or $\lceil W(0)\rceil$, 
where $W(0)$ is the degree-weighted average
of the initial opinions, 
provided that $W(0)$ is separated from the closest 
integer by a positive constant.
The graph parameter $\lambda(G)$ is the second largest eigenvalue (in absolute value) of the transition matrix
of a random walk on $G$.
A results of this nature was shown earlier in~\cite{CRS24},
but under a strong assumption that $K \cdot \lambda(G) = o(1)$.

\subsection{The wider context}
 \label{sec:relatedwork}

In this section we outline related voting processes and opinion dynamics.
The subject has a substantial literature, for more detail see the survey~\cite{DBLP:journals/sigact/BecchettiCN20}.
We remark that in contrast to problems such as Byzantine consensus, voting processes assume honest
and fault-free behavior. The aim  is for the final outcome to be a 
useful function of the initial
opinions rather than one imposed adversarially.

\noindent
\paragraph{\textbf{The pull voter model.}} For  standard  pull voting  on a connected graph~$G$,
a given opinion $S$ wins with probability $d(S)/2m$, where $d(S)$ is the
sum of the degrees of the nodes initially holding opinion $S$;
see, for example, Hassin and Peleg~\cite{DBLP:journals/iandc/HassinP01}.
Berenbrink et al.~\cite{DBLP:conf/icalp/BerenbrinkGKM16}  consider {lazy synchronous pull voting, where in each round each node adopts the sampled opinion with probability $1/2$}, and show that with constant
probability  consensus is reached within $O(\gamma(G) n /\Phi(G))$ rounds.
Cooper and Rivera~\cite{CooperR16}  prove the corresponding bound of $O(\gamma(G)n^2/\Phi(G))$ on the number of steps in asynchronous pull voting with two opinions.

For general bound of $O((1/(1-\lambda(G)))\cdot (\log^4 n+\rho(G)))$
on the expected convergence time of the synchronous pull voting, see
\cite{DBLP:journals/siamdm/CooperEOR13}.
Here 
$\rho(G)$ is the ratio of the square of the sum of node
degrees over the sum of the squared degrees, which ranges from $O(1)$ (star graph)
to $n$ (regular graphs).
Cooper and Rivera~{\cite{CooperR16}} 
introduces a
\emph{Linear Voting Model}, which
generalizes several models of voting, including
asynchronous and synchronous pull voting and push voting,
and derives bounds on the probability that a given opinion wins and on
the expected voting time.

\noindent
\paragraph{\textbf{Opinion Dynamics.}}
Doerr et al.~\cite{10.1145/1989493.1989516}
consider  the synchronous \protocol{MedianRule} protocol on complete graphs with opinions drawn from
an \emph{ordered} set, say the set $[K-1]$. In each round each node selects \emph{two} random neighboring nodes
and updates its opinion to the median of its own opinion
and the two neighbors' opinions.
The results in~\cite{10.1145/1989493.1989516} include bounds on the convergence time of this process, considering also adversarial scenarios.
Note that in the case of two initial opinions ($K=2$) this protocol reduces to picking the majority of the three opinions.
However, this process does not always converge to a single opinion on general graphs.

Another widely studied consensus protocol
is \emph{\mbox{$j$-Majority} dynamics}.
Every node samples randomly $j$ {neighbors} and adopts
the majority opinion among the sample, breaking ties uniformly at random. 
The variants for $j$=2 and $j$=3 have
been analyzed under the names of \protocol{Two-Choices} dynamics
\cite{DBLP:conf/icalp/CooperER14, DBLP:conf/wdag/CooperERRS15,DBLP:conf/wdag/CooperRRS17, DBLP:conf/podc/ShimizuS25}
and the \protocol{3-Majority} dynamics
\cite{DBLP:journals/dc/BecchettiCNPST17, DBLP:conf/podc/GhaffariL18,
DBLP:conf/podc/BerenbrinkCEKMN17, DBLP:conf/soda/CooperMRSS25, DBLP:conf/podc/ShimizuS25}.
(The 'strict' $2$-Majority dynamics, in which a node randomly chooses between the two selected opinions 
when they differ, is equivalent to pull voting. 
Therefore, \protocol{Two-Choices} dynamics assumes the node retains its current opinion in such a case.)
In \emph{averaging processes}, the opinions have numeric values and
the nodes adopt the average of the opinions of $j$ random neighbors.
Such processes are considered, for example,
in Berenbrink 
et al.~\cite{DBLP:conf/podc/BerenbrinkCGMMR23}, where
tight bounds on the variance on the final opinion for regular graphs are derived.
If the initial opinions do not depend on the number of nodes the variance is
negligible, and hence the nodes are able to estimate the average of the
initial node opinions. Interestingly, this variance does not depend on the graph
structure.

\noindent
\paragraph{\textbf{Voting in social networks.}}
In the social sciences, models of opinion dynamics typically incorporate some notion of
compromise when agents interact, rather than simple adoption of opinions held by others.
This is often expressed by updating opinions toward 
the opinion of the contacted nodes.  
A prominent example is the DeGroot model~\cite{DeGroot1974}, in which agents repeatedly
calculate, and adopt as their own opinion, a weighted average of the opinions of \emph{all} their 
neighbors.
In \emph{bounded-confidence models}, 
agents take into account opinions of other agents 
only when those opinions are not too far away from their own. 
This can be implemented 
either by a step change or incorporating a weighted average
of the other opinions and the original opinion of the node. 
A low threshold on two opinions to be considered as  
not being too far apart may prevent the convergence to a single opinion.
Examples of 
such
models are
Deffuant--Weisbuch~\cite{Deffuant2000} and Hegselmann--Krause~\cite{HegselmannKrause2002}.

\subsection{Additional notation and preliminaries}
\label{subsec:preliminaries}

\ProtAV is a Markov process evolving in discrete steps $t=0,1,2,\ldots$ as specified in
Protocol~\ref{prot:edgeinteraction}.
Label the nodes  $V=\{1, \ldots, n\}$,
and denote the random state of the process at step $t$ by
$\X{t}=(X_1(t),\ldots,X_n(t))$, where $X_i(t)$ is the opinion of node $i$,
and its concrete realization (the configuration at step $t$) by $\vx(t)=(x_1(t),\ldots,x_n(t))$.
For a configuration $\Vec{x} = (x_1, \ldots, x_n)$, let
$\xmax{} \coloneqq \text{max}_{i \in V}\, x_i$ and
$\xmin{} \coloneqq \min_{i \in V}\, x_i$.
The \emph{discrepancy} of $\Vec{x}$ is defined as
$\disc{\Vec{x}}\coloneqq \xmax{} - \xmin{}$. 
Let $\mathcal{X}(K)\coloneqq \{\vx=(x_1,\ldots,x_n): x_i\in[K]\}$
be the set of all possible configurations with values (opinions) in $[K]\coloneqq \{0,1,\ldots,K\}$. Thus if $\Vec{x} \in \mathcal{X}(K)$ then $0 \le \disc{\Vec{x}} \le K$.

The convergence time from $\Vec{x}$ to consensus is a random variable defined as
\begin{align*}{T_G(\Vec{x})}
&\coloneqq
\min\left\{t\ge0: \Vec{X}(0) = \vec{x}, \disc{\X{t}}=0\right\}.
 \end{align*}
Similarly, define the time for the process to reduce the discrepancy  to a specific value $\kappa$ as
\begin{align*}{T_G^{(\kappa)}(\Vec{x})}
&\coloneqq
\min\left\{t\ge0: \Vec{X}(0) = \vec{x}, \disc{\X{t}}\le\kappa\right\}.
 \end{align*}
The worst-case expectations of ${T_G(\Vec{x})}$ and ${T_G^{(\kappa)}(\Vec{x})}$, 
over the initial configurations $\Vec{x}$ with discrepancy at most $K$, are denoted by
\begin{align*}
\mathcal{T}^{(K:0)}_G & \coloneqq \text{max}\big\{\E{T_G(\Vec{x})}:\,{\Vec{x}\in\mathcal{X}(K)}\big\},\\
\mathcal{T}^{(K:\kappa)}_G & \coloneqq \text{max}\left\{\E{T_G^{(\kappa)}(\Vec{x})}:\,{\Vec{x}\in\mathcal{X}(K)}\right\}.
\end{align*}
If   the opinion set is $\{0,1\}$,  then the discrepancy  $K=1$,
and \ProtAV coincides with  standard (2-value) pull voting.
We denote the  worst-case expected  convergence time of pull voting by
\begin{align*}
\mathcal{T}^{2V}_G & \coloneqq \text{max}\Big\{\E{T_G(\Vec{x})}:\,{\Vec{x}\in\mathcal{X}(1)}\Big\}
\;\; \left(\: = \mathcal{T}^{(1:0)}_G\,\right).
\end{align*}

Let $m=|E(G)|$ denote the number of edges of $G$, $d_i$  the degree of node~$i$,  and  $\daveNG = {2m}/{n}$  the average degree.
For a node of degree $d_i$, denote the \emph{degree imbalance}  as $\gamma_i \coloneqq \daveNG/d_i={2m}/{(n\, {d}_{i})}$.
Let $\gamma(G)=\gamma_{\mathsf{max}}(G)\ge 1$ and $\gamma_{\mathsf{min}}(G)\le 1$
as the maximum and minimum of $\gamma_i$ for $i\in V$.
If $\gamma_{\mathsf{max}}(G) = O(1)$ and $\gamma_{\mathsf{min}}(G) = \Omega(1)$,
then   with $\dmaxNG/\dminNG = O(1)$ and we say the graph is \emph{almost regular}.

For a set $S\subseteq V$, its \emph{volume} is $\vol{S}\coloneqq \sum_{i\in S} d_i$.
For two disjoint sets $S,S'\subseteq V$, define
$E(S,S')$ as the set of edges between $S$ and $S'$,
and $\mathbb{O}(S)\coloneqq E(S,V\setminus S)$ as the set of edges leading out of $S$,
also referred to as the set of edges in the \emph{cut} $(S,V\setminus S)$.
For $S \subseteq V$, $S\not\in\{\emptyset, V\}$,
the \emph{conductance of a cut} $(S,V\setminus S)$ in $G$ is defined as 
$\frac{|\mathbb{O}(S)|}{\min\{\vol{S},\vol{V\setminus S}\}}$,
and the \emph{conductance} $\Phi(G)$ of $G$ is defined as the minimum conductance of a cut in $G$.
Thus
\begin{align}
\label{eq:conductance}
\Phi(G)
\coloneqq
\min\left\{\frac{|\mathbb{O}(S)|}{\vol{S}}:\;
S\subseteq V,\: 0 < \vol{S}\le m\right\}.
\end{align}

Let $A$ be the adjacency matrix of $G$ and $D$ the diagonal matrix of node degrees.
The transition matrix 
of a simple random walk on~$G$ 
is
$P = D^{-1}A$.
Let $1=\lambda_1\ge \lambda_2\ge \cdots \ge \lambda_n\ge -1$ be the eigenvalues
of~$P$ and let $\lambda(G) \coloneqq \text{max}\{\abs{\lambda_2}, \abs{\lambda_n}\}$ denote the \emph{second absolute eigenvalue} of $P$.
Then $1-\lambda(G)$ is referred to as \textit{(absolute) spectral gap}. 
An almost regular graph $G$ is called a \emph{$\lambda$-expander}, or simply an \emph{expander},
if $\lambda(G)\le \lambda$, where $1-\lambda$ is at least a positive constant.
If the underlying graph is clear from the context, we may use $\gamma$, $\Phi$ and $\lambda$ as abbreviations of $\gamma(G)$, $\Phi(G)$ and~$\lambda(G)$.

In this paper, 'with high probability'
(w.h.p.) means with probability at least $1-O(1/n^c)$ for
some constant $c>1$. 
The probability of an event $\mathcal{A}_n$ is $p + o(1)$, or $p - o(1)$, if 
$\lim\limits_{n\rightarrow\infty}\pr{\mathcal{A}_n} = p$.

\section{The results of this paper}
\label{sec:results}

In this section we give the  bounds on the convergence time of DIV obtained in this paper. These and the  relevant 
previous bounds are summarized in \autoref{tab:relatedwork}.

\subsection{A general upper bound for DIV based on conductance}

We begin by establishing a general upper bound for the convergence time of DIV on any connected graph $G$.
This bound is given in terms of the conductance $\Phi(G)$ and the expected
convergence time of $2$-Value Pull Voting $\mathcal{T}^{2V}_{G}$.
\begin{theorem}
\label{thm:main}
Consider the \ProtAV process on a graph $G$ with conductance $\Phi(G)$. Assume
$\vec{x}$
is an arbitrary configuration with discrepancy $K\ge 1$. The following two
bounds on the convergence time ${T_G(\vec{x})}$
hold, with bound~\eqref{boundwhp} holding
both in expectation and with high probability (\whp),
\begin{align}
{T_G(\vec{x})}
&\;\in\;
O\!\left(
\frac{ K n \log(Kn)}{\Phi(G)^2}
+
\frac{\log(n)}{\Phi(G)} \cdot \mathcal{T}^{2V}_{G}
\right),\label{boundwhp}\\
\E{T_G(\vec{x})}
&\;\in\;
O\!\left(
\frac{ K n \log(Kn)}{\Phi(G)^2}
+ \frac{\gamma(G) n^2}{\Phi(G)^2}\right).
\label{boundonexpectation}
\end{align}
\end{theorem}
When the discrepancy is $1$, DIV coincides with standard 2-value pull voting, and the worst-case expected convergence time is  $\mathcal{T}^{2V}_G$.
A  bound of $\mathcal{T}^{2V}_{G}=O(n m/d_{\min}\Phi(G))=O(\gamma(G)n^2/\Phi(G))$, 
is from~\cite{CooperR16}.
A general bound of $O\big(K\mathcal{T}^{2V}_{G}\big)=O(Kn^2\gamma(G)/\Phi(G))$
on the expected time to consensus in DIV is given in \cite{CRS23}.
The above theorem refines this by separating 
the cost of reducing the initial, and potentially large, discrepancy~$K$ to 
$O((\log n)/\Phi(G)$
(the first term in bounds~\eqref{boundwhp} and ~\eqref{boundonexpectation})
from the cost of reducing the discrepancy
when it is already relatively small but increasingly 
harder to reduce further.

For the expected consensus time, 
the overall effect of this refinement 
is that the product $K\cdot n\gamma(G)$ in the previous general bound 
$O(Kn^2\gamma(G)/\Phi(G))$ (after factoring out $n/\Phi(G)$)
is replaced in the new bound~\eqref{boundonexpectation}
with the sum $K\cdot\frac{\log(Kn)}{\Phi(G)} + n\gamma(G)\cdot\frac{1}{\Phi(G)}$, 
offering substantial improvements for wide
ranges of parameters.
For example, for regular graphs
($\gamma(G) = 1$) with constant conductance,
bound \eqref{boundonexpectation} 
becomes $O(Kn{\log(Kn)} + n^2)$, while the general bound from \cite{CRS23} is $O(Kn^2)$,
so the new bound is asymptotically better
whenever $K \in \omega(1)$ and $\log(Kn)=o(n)$.

\subsection{A general lower bound based on conductance}
There are regular graphs with an edge cut whose sides each contains half the nodes and whose conductance equals
the conductance of the graph $\Phi(G)$. 
In this regime,
we next show that dependence 
of the convergence time on both the graph conductance and 
the initial discrepancy $K$ is unavoidable in the worst case. 
Each DIV interaction only changes a single opinion by at most one. If two large groups start with opinions $K$ and $0$, the edges between them limit how quickly their average opinions can approach one another. 

\begin{theorem}[Lower Bound]
\label{thm:lowerbound}
Let $1 \le q_n \le O(n)$. 
There are regular graphs $G$ with conductance $\Phi(G) =  q_n/n$, such that for each $K\ge1$, there is an initial configuration~$\Vec{x}$ with discrepancy  $K$ for which
$\E{T_G(\Vec{x})} \in \Omega\!\left({K n}/{\Phi(G)} + \mathcal{T}^{2V}_G\right)$.
\end{theorem}
The two terms identify distinct sources of slow convergence: a large initial gap across a sparse cut, and the worst-case completion time for 2-value pull voting.
For regular graphs with constant conductance satisfying Theorem \ref{thm:lowerbound},  when $K\ge n$, the upper bound \eqref{boundonexpectation} of Theorem \ref{thm:main} and the lower bound of Theorem \ref{thm:lowerbound} agree up to logarithmic factors. 
For non-constant conductance, 
the discrepancy-dependent part of 
the upper bounds of Theorem \ref{thm:main} has an additional factor $1/\Phi(G)$. 
For path graphs, where the set of edges is 
$E = \{ \{1,2\}, \{2,3\}, \ldots, \{n-1,n\}\}$, a lower bound of $\Omega(n^3)$ was shown for $K=2$ in~\cite{CRS23}.

\subsection{A new upper bound for expander graphs}

For graphs with  expansion quantified by $\lambda(G)$, we obtain 
the following bounds on the time $T_G^{(2)}$ to reduce the discrepancy to~$2$,
that is, on the time to reach the stage where only 3 consecutive opinions remain 
(say, opinions $0,1,2$, after shifting the opinions to align the smallest remaining opinion with $0$).
\begin{theorem}
\label{thm:main2}
Consider the \ProtAV process on a graph $G$
with $\gamma(G) = O(1)$ and the second-largest absolute eigenvalue 
$\lambda(G) \le \xi<1$, for a constant $\xi$.
Assume $\vec{x}$ is an arbitrary configuration with discrepancy $K$. 
Then we have the following bounds on the time needed to reduce the discrepancy to $2$: 
\begin{align}
{T_G^{(2)}(\vec{x})}
&\;\in\; O\!\left(K n \log (Kn) + \lambda(G)^2\, n^2 \log n + n \log^2 n \right)\; \text{\whp}
\label{Thm2.3-bound}\\
\E{T_G^{(2)}(\vec{x})}
&\;\in\; O\!\left( K n \log (Kn) + \lambda(G)^2\, n^2 \log n \right),
\label{Thm2.3-bound2}
\end{align}
implying 
\begin{align}\label{Thm2.3-boundX}
\mathcal{T}^{(K:0)}_G
\;\in\; O\!\left(K n \log (Kn) + \lambda(G)^2\, n^2 \log n + \mathcal{T}^{2V}_{G} \right).
\end{align}
\end{theorem}
Intuitively, the smaller the discrepancy, the longer it takes to reduce it by another unit,
as mentioned earlier in the context of \autoref{thm:main}.
This came to light also in previous results on DIV presented in~\cite{CRS24},  such as
\begin{equation}\label{T-val}
\mathcal{T}^{(K:2)}_G
\in\, O\left(K  n\log n + \lambda(G) K n^2  \right)
\;\;\; \text{and} \;\;\; 
\mathcal{T}^{(2:1)}_G
\in\, O\left(n^{5/3} \log n +\sqrt{\lambda(G)}\, n^2 \right).
\end{equation}

Similarly as before in \autoref{thm:main}, the bounds in \autoref{thm:main2} 
separate the cost of reducing the discrepancy from~$K$ to $O(\log n)$, 
which is now $O(Kn\log(Kn))$ since $\lambda(G) \le \xi<1$ implies $\Phi(G) = \Omega(1)$,
from the cost of the next phase when 
the discrepancy is reduced from $O(\log n)$ to $2$,
which is $O(\lambda(G)^2n^2\log n)$. 
The additional term of $O(n\log^2 n)$ 
in~\eqref{Thm2.3-bound}, 
which is subsumed by the $O(Kn\log(K n))$ term 
for $K=\Omega (\log n)$,
ensures that bound~\eqref{Thm2.3-bound}
holds with high probability.
The final phase of DIV,
from discrepancy 2 to consensus, contributes an additional $O(\mathcal{T}^{2V}_{G})$ to the upper bound in~\eqref{Thm2.3-boundX}.

The crucial contribution of \autoref{thm:main2}
is that it replaces the term $\lambda(G) K n^2$
in the first bound in \eqref{T-val}
with $\lambda(G)^2 n^2 \log n$ 
in bound~\eqref{Thm2.3-bound2}.
As shown in  \autoref{thm:finalvalue} below,
this  improvement leads to much better bounds
on the concentration of the final value.

\begin{table}
\centering
\setlength{\tabcolsep}{2pt}
\renewcommand{\arraystretch}{1.4}
\begin{threeparttable}
\begin{tabularx}{\linewidth}{@{}
  >{\raggedright\arraybackslash}p{0.2\linewidth}
  >{\raggedright\arraybackslash}p{0.02\linewidth}
>{\raggedright\arraybackslash}X
  >{\raggedright\arraybackslash}p{0.3\linewidth}  
 >{\raggedright\arraybackslash}p{0.095\linewidth}@{}}

\toprule
{\textbf{Graph class}} & &
{\textbf{Expected steps}} &
{\textbf{Assumptions}} & \textbf{Reference} \\
\midrule
\multicolumn{5}{@{}l@{}}{{\underline{\textbf{Previous results}}}} \\
\addlinespace[0.4ex]
Any connected $G$ & &
$\mathcal{T}_G^{(1:0)}=\mathcal{T}^{2V}_G=O(\gamma n^2/\Phi)$ & {---} & \cite{CooperR16} \\
 & &
$\mathcal{T}_G^{(K:0)}=O\big( K\cdot\mathcal{T}^{2V}_G\big) =O(K\gamma n^2/\Phi)$ &
{---} & \cite{CRS23} \\
\addlinespace[0.7ex]
$\mathcal{G}_{n,p}$ &  & $\mathcal{T}_G^{(K:0)}=O(n^2)$ &
{$np\ge\log^{1+\epsilon}n, K=o(\log n)$}& \cite{CRS23} \\
\addlinespace[0.7ex]
$\lambda$--expander &  &
$\mathcal{T}_G^{(K:2)}=O(Kn\log n+K\lambda n^2)$ &
{$\gamma=O(1)$;\;\; $\lambda=o(1)$} & \cite{CRS24} \\
 &  &
$\mathcal{T}_G^{(2:1)}=O(\sqrt{\lambda}n^2+n^{5/3}\log n)$ & {$\gamma=O(1)$;\;\; $\lambda=o(1)$} & \\
\addlinespace[0.7ex]
Complete graph $K_n$ &  & $\mathcal{T}_G^{(K:2)}=O(Kn\log n)$ & {---} & \cite{CRS24} \\
 &  & $\mathcal{T}_G^{(2:1)}=O(n^{5/3}\log n)$ &--- & \\
\addlinespace[1.2ex]
\multicolumn{5}{@{}l@{}}{{\underline{\textbf{This paper}}}} \\
\addlinespace[0.4ex]
Any connected $G$ &  & {$\mathcal{T}_G^{(K:O(\log n/\Phi))}=O\!\left({Kn\log(Kn)}/{\Phi^2}\right)$} &
{---} & {Cor.~\ref{thm:technical}} \\
&  & {$\mathcal{T}_G^{(K:0)}=O\!\left(({Kn\log(Kn)+\gamma n^2})/{\Phi^2}\right)$} &
{---} & Thm~\ref{thm:main} \\
\addlinespace[0.7ex]
$\lambda$-expander &  &
{$\mathcal{T}_G^{(K:2)}=O(Kn\log(Kn)+\lambda^2n^2\log n)$} &
{$\gamma=O(1)$, $\lambda\le c<1$} & Thm~\ref{thm:main2} \\
\addlinespace[0.7ex]
{Regular graph} &  &
{$\mathcal{T}_G^{(K:0)}=\Omega\big(Kn/\Phi+\mathcal{T}^{2V}_G\big)$} &
{---} & Thm~\ref{thm:lowerbound} \\
\bottomrule
\end{tabularx}
\end{threeparttable}
\caption{A summary of results giving the expected number of steps to consensus (or near consensus)
for asynchronous DIV on an $n$-node graph $G$. 
}
\label{tab:relatedwork}
\end{table}

\subsection{{Concentration of the consensus opinion on expander graphs}}

Finally, and perhaps most importantly, we consider the value the process converges to.
For a vector of opinions $\vec x$, let $W(\vx)\coloneqq \frac{1}{2m}\sum_{i=1}^n d_i x_i$ denote the degree-weighted average opinion, or the mean opinion for short.
As the quantity $W(t) = W\big(\X{t}\!\big) = \frac{1}{2m}\sum_{i=1}^n d_i \cdot X_i(t)$ is a martingale,
we have $\E{W(T_G(\vec{x}))} = W(0)$, where $\vec{x}$ is the vector of initial opinions,
$W(0) = W(\vx)$ is the initial mean opinion, and $W(T_G(\vec{x}))$ is the 
final consensus opinion -- the winning opinion 
-- since $T_G(\vec{x})$ is the time when all nodes finally have the same opinion.
That is, the expected value of the winning opinion is equal to the initial mean opinion. 

As the winning opinion is integer, the question how well its value is
concentrated around the mean can be phrased as the question how probable it is that 
the winning opinion is an integer other than the initial mean rounded up or rounded down.
Our answer to this question, for expanders $G$ with sufficiently small $\lambda(G)$, 
is formalized in the following theorem, 
obtained by combining \autoref{thm:main2} with the analysis of $\mathcal{T}^{(2:1)}_G$ from~\cite{CRS24}
(the second bound in~\eqref{T-val}).

\begin{theorem}
\label{thm:finalvalue}
Let $G$ be an expander
with second-largest absolute eigenvalue $\lambda(G) \in o(1/\log^4 n)$.   
Consider the \ProtAV process starting from an arbitrary configuration with discrepancy 
$K \in o(n/\log^2 n)$ and
the initial degree-weighted average of the node opinions
$W(0)$. 
Let $c_1$ and $c_2$, where $0 \le c_1 \le c_2 \le 1/2$, bound 
the distance of $W(0)$ to the nearest integer, i.e.,
\[ c_1 \; \le \; \min{}\!\big\{W(0) -\lfloor W(0) \rfloor, \lceil W(0) \rceil -W(0)\big\} 
\; \le \; c_2.
\]
\begin{enumerate}
\item \label{Thm2.3-item1}
If $c_1$ is a positive constant, then \whp the consensus value of the process, 
is either $\lceil W(0) \rceil$ or $\lfloor W(0) \rfloor$. Moreover 
the consensus value is $\lceil W(0) \rceil$ 
with probability $W(0) - \lfloor W(0) \rfloor +o(1)$,
and $\lfloor W(0) \rfloor$ with probability $\lceil W(0) \rceil -W(0) +o(1)$.
\item \label{Thm2.3-item2}
If  $c_2=o(1)$, let  $\lfloor W(0) \rceil$  denote the nearest integer to $W(0)$.
Then \whp the consensus value is 
in $\{ \lfloor W(0) \rceil-1, \lfloor W(0) \rceil, \lfloor W(0) \rceil+1\}$, 
and it is $\lfloor W(0) \rceil$ with probability $1-o(1)$.
\end{enumerate}
\end{theorem}

For a one-sided process in which nodes adjust their values based only on observing that of a neighbor, it seems remarkable that the  weighted average is conserved. Indeed the theorem describes a two point concentration of the weighted average. 

If $W(0)$ lies  between two consecutive integers and is a positive constant distance from each, 
then with high probability the consensus value 
is one of these integers.
Moreover the original average $W(0)$ determines their respective winning probability 
(part (1) of the theorem).

If $W(0)$ is within $o(1)$ to the nearest integer, 
denoted by $\lfloor W(0) \rceil$, then with high probability
the consensus value is in 
$\big[\lfloor W(0) \rceil -1, \lfloor W(0) \rceil+1\big]$,
and it is exactly $\lfloor W(0) \rceil$
with probability $1-o(1)$.

Theorem 1 of~\cite{CRS24} gives a result of this type 
but under a strong assumption that $K\lambda=o(1)$.  
Our assumption that $\lambda= o(1/\log^4n)$ and $K=o(n/\log^2 n)$ removes the requirement that $K\lambda=o(1)$, 
substantially increasing the combine range 
of the parameters $\lambda$
and $K$ for which the strong concentration of the 
consensus value can be proven.

\section{Outline of the proofs of our results}
\label{sec:outlines-of-proofs}

\subsection*{Proof of \autoref{thm:main}}

We reduce the discrepancy in two phases, first to below $\Omega(\log n/\Phi(G))$ using a potential argument, and then  to consensus.
With this in mind, let $\alpha=32\log n/\Phi(G)$. 
For a discrepancy $K\ge\alpha$, translate the opinion range into $[0,K]$ and 
consider (integer) \emph{thresholds} $k$ between $K/2+1$ and $3K/4$, 
observing the rate at which the support for opinions outside of the central region of
$[K-k, k]$ diminishes (note that $K/4 \le K-k \le K/2 \le k \le 3K/4$).
More precisely, for each threshold, 
we track the sum of the degree-weighted squared 'excesses' of opinions above $k$, namely
$\sum_{i\in V}d_i \cdot (\text{max}\{x_i(t) - k, 0\})^2$, and the analogous sum for the 
'deficits' of opinions below $K-k$, 
namely $\sum_{i\in V}d_i \cdot (\text{max}\{(K-k) - x_i(t), 0\})^2$.
We refer to these two sums as, respectively, the positive and negative potentials w.r.t.\ $k$.
If either of them vanishes, the discrepancy is at most $k$.

The main task is to find a threshold $k$ 
for which one of the two potentials decreases sufficiently quickly.
Considering the positive potential for a threshold $k$, 
edges joining (nodes with) opinions which are at least $k+1$ to opinions which are at most $k$ (bypassing opinion $k$) contribute to the expected decrease of this potential.
We call a threshold $k$ good if the set of nodes with opinion at least $k$ has volume at most $m=|E(G)|$ and sufficiently many edges bypass nodes of opinion~$k$, where 'sufficiently many'
is parameterized by the conductance of the graph. 
A counting argument shows that some threshold must be good for a constant fraction of the steps, 
either w.r.t.\ the positive potential or, using analogous definitions, w.r.t.\ 
the negative potential.
The resulting drift of the corresponding potential forces it to vanish with high probability 
within $O(Kn\log(Kn)/\Phi(G)^2)$ steps, reducing the discrepancy to at most $3K/4$.
Repeating at successively smaller discrepancies gives a geometric sum of time bounds,
yielding a time bound of $O(Kn\log(Kn)/\Phi(G)^2)$, both in expectation and with high probability, to reduce discrepancy to $\alpha=32\log n/\Phi(G)$.
The details are  presented in \Autoref{sec:prooftechnical}, specifically \Autoref{lemma:phasescombined} and \Autoref{thm:technical}.

In the second phase, when the discrepancy reduces from $\alpha$ to ~$0$, to show bound~\eqref{boundwhp},
we use the reduction to $2$-value pull voting from~\cite{CRS23}.
Each decrease of the discrepancy by one takes expected time at most $\mathcal{T}^{2V}_{G}$, so the expected remaining time is at most $\alpha\mathcal{T}^{2V}_{G}$.
The high-probability bound is obtained by dividing time into blocks of $2\mathcal{T}^{2V}_{G}$ steps. Each block starting with positive discrepancy reduces it with probability at least $1/2$.
Since $\alpha=\Omega(\log n)$, a binomial tail bound implies that with high probability
within $O(\alpha)$ blocks the discrepancy is reduced to $0$.

By combining the analysis of the two phases as outlined above, we prove bound~\eqref{boundwhp}.
The sharper bound~\eqref{boundonexpectation} on the expected completion time follows from a more refined analysis of the second phase, developed in \Autoref{Sec-Thm2.1-2nd-bound}.

\subsection*{Proof of \autoref{thm:lowerbound}}

For a lower bound on consensus time, we use regular graphs with a  balanced cut $(U,V\setminus U)$ (i.e., with $|U|=n/2$), that attains the graph's conductance. 
Such graphs exist for conductances ranging from $\Theta(1/n)$ to $\Theta(1)$; the construction is given in \autoref{sec:lowerboundproof}.
We assign  opinion $K$ to every node in $U$, and  opinion $0$ to every node outside $U$.
The difference between the two groups' average opinions starts at $K$ and must reach zero for consensus.

On a regular graph, interactions within either group have zero expected effect on its average opinion.
Only interactions across the cut can produce drift in the difference between the two averages.
Such interactions occur with probability $|E(U,\overline U)|/|E|\le\Phi(G)$, and a single update changes the difference by at most $1/|U|=O(1/n)$.
Thus its one-step drift has absolute value $O(\Phi(G)/n)$.
A martingale argument with bounded increments then shows that, with constant probability, the gap is still positive after $\Omega(Kn/\Phi(G))$ steps, giving the same lower bound on the expected convergence time.

\subsection*{Proof of \autoref{thm:main2}}
We use the following approach to reduce the discrepancy to 2 in expanders.
The assumptions that $\gamma(G)=O(1)$ and $\lambda(G)\le \xi <1$ allow a sharper analysis once the discrepancy is small.
As Cheeger's inequality (\autoref{Cheeger}) gives $\Phi(G)\ge(1-\lambda(G))/2=\Omega(1)$,
referring back to the outline of the proof of \Autoref{thm:main} given earlier,
any initial discrepancy above 
$\alpha=32\log n/\Phi(G) = O(\log n)$ is reduced to at most $\alpha$ 
in $O(Kn\log(Kn))$ steps, both in expectation and with high probability.
Write $L\le\min\{K,\alpha\}=O(\log n)$ for the remaining discrepancy.

The key to the next stage is the property that 
whenever the discrepancy is at least $3$, it decreases by~$1$ 
with probability at least $1/3$ within
$
\beta=O\left(\lambda(G)^2n^2+n\log n\right)
$
steps (\autoref{lem:one-step}).
For this one-step argument, write $K\ge3$ for the current discrepancy, 
translating the opinions into $[0,K]$ if needed, 
and track the volumes of the support for opinions $0$ and $K$.
The argument depends on whether the supports for these two extreme opinions are large or small.

If both support volumes exceed $32m\lambda(G)^2$, the expander mixing lemma shows that 
threshold $K-1$ is good w.r.t.\ the positive or negative potential.
Thus, if both supports stay large throughout a sequence of $T=O(n\log n)$ steps, 
this threshold is good in at least half of the steps w.r.t.\ the positive potential,
or in at least half of the steps w.r.t.\ the negative potential.
Then the potential argument used in the proof of
\Autoref{thm:main} can be adapted to derive 
a bound on the probability 
that no extreme opinion disappears during the block.
Consequently, with high probability, 
by time $T$ either the discrepancy has dropped or 
the support of one of the extreme opinions has become small.

When the support for an extreme opinion $0$ or $K$ 
has volume $s\le32m\lambda(G)^2$ at some step $t$,
we show, by coupling to $2$-value pull voting, 
that with probability at least $1/2$,
this extreme opinion disappears,
hence the discrepancy decreases,
within additional
$T' = O\left(\frac{sn}{\Phi(G)\dmin}\right)$ steps
(\autoref{lem:dropbyone}),
which is
$O\left(\lambda(G)^2n^2\right)$
since $\Phi(G)=\Omega(1)$ and 
$\dmin=\dave/\gamma(G) = \Omega(m/n)$.
Combining the two alternatives gives the claimed probability of at least $1/3$ that within $\beta= T+T'$ steps
the discrepancy decreases by $1$.

This success probability holds uniformly over all configurations with discrepancy at least $3$, so we can repeat the argument in blocks of $\beta$ steps.
Geometric waiting-time and binomial tail bounds give $O(L)$ blocks in expectation and $O(L+\log n)$ blocks with high probability to reach discrepancy at most $2$
(\autoref{lem:k-step}).
Since $L=O(\log n)$, the corresponding time bounds are
$O(L(n\log n + \lambda(G)^2n^2))$ and
$O(n\log^2 n + \lambda(G)^2n^2\log n)$.
Adding to this the $O(Kn\log(nk))$ bound on the number 
of steps to 
reduce the discrepancy to $O(\log n)$ and recalling
that $L\le K$
give bounds~\eqref{Thm2.3-bound2} and~\eqref{Thm2.3-bound},
respectively.

Finally, the reduction to $2$-value pull voting from~\cite{CRS23} bounds the expected time 
to reduce the discrepancy from $2$ to consensus by $O(\mathcal{T}^{2V}_{G})$.
Adding this to~\eqref{Thm2.3-bound2} gives~\eqref{Thm2.3-boundX}.

\subsection*{Proof of \autoref{thm:finalvalue}}

This theorem
follows directly from \autoref{thm:main2},
the properties of 2-value pull voting given in \cite{DBLP:journals/iandc/HassinP01}, 
and the martingale concentration argument
the degree-weighted mean $W(t)$
from \cite{CRS23,CRS24}, as explained below.

Consider a configuration with positive discrepancy.
First, \autoref{thm:main2} bounds the time to reduce the discrepancy from $K$ to at most $2$.
We combine this with the bound from~\cite{CRS24} for the time from discrepancy $2$ to $1$
(the second bound in~\eqref{T-val}).
Under the assumptions $K=o(n/\log^2 n)$ and $\lambda(G)=o(1/\log^4 n)$, the combined time is $o(n^2/\log n)$ with high probability.
Write $T^{(1)}$ for the first time the discrepancy is at most one.
The degree-weighted mean $W(t)$ is a martingale, and each update changes it by $O(1/n)$ on an almost regular graph.
A concentration result from~\cite{CRS23} 
therefore gives
$|W(T^{(1)})-W(0)|=o(1)$
with high probability
(cf.\  \cite[Theorem 5]{CRS23} or \cite[Lemma 3]{CRS24}).
Thus the process reaches its binary stage 
(2-value pull voting)
before the mean has changed significantly.

At discrepancy one, the remaining opinions are two consecutive integers, and each wins with probability equal to its share of the total degree
\cite{DBLP:journals/iandc/HassinP01}.
If $W(0)$ is a positive constant distance from the integers, then with high probability the two remaining opinions are $\lfloor W(0)\rfloor$ and $\lceil W(0)\rceil$,
and their shares of the total degree 
are $1-x$ and $x$, respectively,
where $x$ is the fractional part of $W(T^{(1)})$,
which is, up to $o(1)$, equal to
the fractional part of $W(0)$.

If, with growing $n$, 
$W(0)$ approaches an integer, that integer has degree share $1-o(1)$ at the binary stage, 
so it wins with probability $1-o(1)$; the consensus lies within one unit of it with high probability.

\section{Proof of Theorem~\ref{thm:main}: Reducing Discrepancy from \texorpdfstring{$K$ to $O({\log n}/{\Phi(G)})$}{K to O(log n/Phi(G))}}
\label{sec:prooftechnical}

\Autoref{lemma:phasescombined} and its \Autoref{thm:technical}, stated and proved in this section,
bound the time needed to reduce the discrepancy from $K$ 
to $O\left({\log n}/{\Phi(G)}\right)$.
\Autoref{lemma:phasescombined} is the main ingredient of the proof of \autoref{thm:main}, 
and one of the main technical contributions of this paper. 

To develop intuition, consider configurations with large discrepancies.
In this regime, the process behaves similarly to
the following discrete load-balancing, or token distribution, process.
Suppose that in each interaction a pair of adjacent nodes $(u,v)$ is selected. 
Then the node with higher load sends one unit of load to the node with lower load.
The expected evolution of load differences between neighboring nodes closely resembles 
the evolution of opinion differences in our process.
In fact, this load balancing process  can be obtained by letting both endpoints of the interaction
perform the update of 
our DIV protocol simultaneously.
The key difference between the two processes is that in our setting 
the average of the opinions changes over time 
whereas the two-sided-update  preserves the total and average load.
Nevertheless, our analysis of DIV shows that techniques developed to analyze load balancing processes, 
most notably potential functions based on squared deviations, remain applicable, 
provided the potential function framework is defined in an appropriate way.
\begin{proposition}
\label{lemma:phasescombined}
Consider the \ProtAV process on  a simple connected graph with conductance $\Phi(G)$.
Assume $\vec{x}(0)$ is an arbitrary configuration with discrepancy
$K \ge \tfrac{32 \log n}{\Phi(G)}$ and $c>0$ is a constant.
Then with probability at least $1-O\left(\tfrac{1}{K^{c+1} n^{c+2}}\right)$,
within
{$T_K = \left\lceil20(c+4)\tfrac{K n \log(Kn)}{\Phi(G)^2}\right\rceil$}
steps, the discrepancy is at most $(\nicefrac{3}{4}) K$.
This immediately implies 
(by the argument of restarting on failure)
that the expected time to reduce the discrepancy to $(\nicefrac{3}{4})K$
is $O\left(\tfrac{K n \log(Kn)}{\Phi(G)^2}\right)$.
\end{proposition}
The following corollary from \Autoref{lemma:phasescombined} gives the first part of the proof of \autoref{thm:main}:
the bound on the time required to reduce the discrepancy from $K$ to $O(\log n / \Phi(G))$.
\begin{corollary}
\label{thm:technical}
Let  $\alpha = \tfrac{32 \log n}{\Phi(G)}$.
Consider the \ProtAV process on a graph $G$ with conductance\/ $\Phi(G)$.
Assume the initial configuration $\vec{x}(0)$ has discrepancy $K\ge\alpha$.
Then, both in expectation and with high probability, the discrepancy is reduced below $\alpha$ in $O\!\left(\tfrac{K n \log(Kn)}{\Phi(G)^2}\right)$ steps.
\end{corollary}
\begin{proof}
Divide the execution of DIV into $O(\log K)$ phases.
Each phase corresponds to reducing the discrepancy by a factor of $\nicefrac{3}{4}$, 
until the discrepancy drops below $\alpha$.
For $i \geq 1$, define  $K_i = (\nicefrac{3}{4})^{i-1}K$. Hence, in phase $i\ge 1$  
the discrepancy is reduced from $K_{i}$ to $K_{i+1}$.
By \autoref{lemma:phasescombined}, the time $T_{K_i}$ 
to reduce the discrepancy from $K_{i}\ge \alpha$ to $(\nicefrac{3}{4})K_i = K_{i+1}$ is
$O\left(\frac{K_i \cdot n \log{(K_i n)}}{\Phi(G)^2}\right)$, 
in expectation and with high probability, 
even when started from the worst possible configuration with discrepancy $K_{i}$.
Since $\sum_{i=1}^{i^\star} K_i = O(K)$, where $i^\star = O(\log(K/\alpha))$ is the 
index of the last phase, then the time $\sum_{i=1}^{i^\star} T_{K_i}$ to reduce the discrepancy from 
$K$ to below $\alpha$ is 
$O\left(\frac{K n \log{(K n)}}{\Phi(G)^2}\right)$,
both in expectation and with high probability 
(the \whp bound follows by the union bound over the $O(\log (K/\alpha))$ phases).
\qedhere
\end{proof}

We now outline the proof of \autoref{lemma:phasescombined}.
The proof is based on a potential-function argument.
However, instead of using a single potential, we introduce a family of potential functions
$\Psi_k(\Vec{X}(t))$, one for each index (threshold) $k \in [K]$.
We define first 
the positive and negative distances w.r.t.\ $k$ 
as, respectively,
\begin{align}
    \varphi^{(+)}_k(x)  &:= \begin{cases}
        x - k, & \textit{if } x \geq k,\\
         0, & \textit{otherwise;}
    \end{cases}
    \nonumber \\
    \varphi^{(-)}_{K-k}(x)  &:= \begin{cases}
        (K-k) - x, & \textit{if } x \leq K-k,\\
         0, & \textit{otherwise}.
    \end{cases}
    \label{eqn:posneqdistance}
\end{align}
Now we define the positive and negative potentials w.r.t.\ a fixed index $k\in [K]$ 
as follows.
Recall that $d_i$ is the degree of node $i$. 
\begin{align*}
 \Psi^{(+)}_k(\X{t}) &= \sum_{i \in V} d_i \cdot \left(\varphi^{(+)}_k(X_i(t))\right)^2, \\
 \Psi^{(-)}_{K-k}(\X{t}) &= \sum_{i \in V} d_i \cdot \left(\varphi^{(-)}_{K-k}(X_i(t))\right)^2.
\end{align*}
The potential corresponding to index $k$ is defined as
\begin{align*}
 \Psi_k(\X{t}) = \min \left\{\Psi^{(+)}_k(\X{t}), \Psi^{(-)}_{K-k}(\X{t}) \right\}.
\end{align*}
Note that the values of positive and negative distances and the positive and negative potential 
functions are always non-negative.
As soon as $\Psi_k(\X{t}) = 0$, the discrepancy is at most $k$
(i.e., has decreased  at least by $K-k$) because all node values are at most
$k$ or all of them are at least $K-k$.
This is stated as \autoref{obs:ActualMirrorred}:(\ref{zeropotzerodisc}) below.

For the analysis, it will be convenient 
to also consider a mirrored process $(\Vec{X}'(t))_{t \geq 0}$,
defined by 
\[
    \Vec{X}'(t) \coloneqq (K-X_1(t), \ldots, K-X_n(t)).
\]
The \emph{original process} $(\Vec{X}(t))_{t \ge 0}$ and \emph{mirrored process} $(\Vec{X}'(t))_{t \ge 0}$ 
are directly coupled: they follow the same node selections in each step of the execution of 
Protocol~\ref{prot:edgeinteraction}.
The {original process}  starts in $\Vec{x}(0)$ and the mirrored process in $\Vec{x'}(0)$.
Under this coupling, the following observations \ref{obs:ActualMirrorred}:(\ref{firstobs})-(\ref{org:mirr}) 
hold.

\begin{observation}\label{obs:ActualMirrorred}
Let $(\Vec{X}(t))_{t \ge 0}$ and $(\Vec{X}'(t))_{t \ge 0}$ be the original and the mirrored processes.
Then for all $t \geq 0$ and $k/2 < k \le K$,
\begin{enumerate}
    \item $\disc{\Vec{X'}(t)} = \disc{\Vec{X}(t)}$,\label{firstobs}
    \item $\Psi^{(+)}_{k}(\Vec{X'}(t)) = \Psi^{(-)}_{K-k}(\Vec{X}(t))$,\label{secondobs}
    \item \label{org:mirr}
    $\Psi_k(\X{t}) = \min \left\{\Psi^{(+)}_k(\X{t}), \Psi^{(+)}_{k}(\Vec{X'}({t})) \right\}$,
    \item $\Psi_{k}(\Vec{X}(t)) = 0 \Rightarrow \disc{\Vec{X}(t)} \leq k$.\label{zeropotzerodisc}
\end{enumerate}
\end{observation}

Informally,
in the original process, the positive potential w.r.t. $k$ reflects the strength of support
for opinions greater than $k$,
while the negative potential reflects the strength of support
for opinions smaller than $K-k$.
The latter is equivalent to the positive potential in the mirrored process reflecting the strength
of support (in the mirrored process) for opinions greater than $k$ 
(\autoref{obs:ActualMirrorred}:(\ref{secondobs})).
Therefore, while the aim is 
to track the changes of the overall potential $\Psi_k(\X{t})$ in the original process,
it suffices to analyze only the positive potential functions, for  the original and
the mirrored processes (\autoref{obs:ActualMirrorred}:(\ref{org:mirr})).
Thus from now on, if we write simply 'potential', we mean the positive potential.

We  show that there is an index $k \in [K/2+1,(\nicefrac{3}{4}) K]$\footnote{Regarding 
the arithmetic on numbers which refer to indices and discrepancies,
we assume rounding down. For example, $K/2$ and $(\nicefrac{3}{4})K$ should be understood as
$\lfloor K/2\rfloor$ and $\lfloor(\nicefrac{3}{4})K\rfloor$, respectively.}
such that within the time $T_K$ defined in the statement of \Autoref{lemma:phasescombined},
\whp\ the potential w.r.t.\ $k$ in the original or mirrored process 
becomes zero, and hence the discrepancy reduces to at most $(3/4)K$.

In \Autoref{sec:potentialdrop}, we prove two key properties of the potential functions.
First, we show that the expected value of 
each potential function is monotonically non-increasing over time.
Second, we show that 
for each index $k\in [1, K-1]$ which is 'good' in the current step,
that is, informally, for each index $k$ 
such that in the current step, the number of edges between nodes 
with opinions at least $k+1$ and nodes with opinions at most $k-1$ is relatively large,
the potential w.r.t.\ $k$ decreases in this step, in expectation,
by an amount proportional to the graph conductance.

\begin{definition}[$k$-good configurations; good indices]
\label{def:kgood}
For any configuration $\vec{x}$ and integer $k$, define
\begin{align}
    S   \coloneqq \{ i \in V \mid x_i \ge k \},
    \;\;
    S_{+1} \coloneqq \{ i \in V \mid x_i \ge k+1 \},
    \;\;
    \overline{S} \coloneqq V \setminus S.
\label{defn-Splus}
\end{align}
We call $\vec{x}$ a $k$-good configuration, if and only if,
\begin{equation}\label{goodk}
\vol{S} \le m=|E(G)|, \quad \text{ and } \quad \abs{E(S_{+1},\ovl S)} \geq \tfrac{1}{4}\Phi(G)\cdot\vol{S_{+1}}.
\end{equation}
Further, if the configuration $\vx(t)$ at the beginning of step $t$ in the original process
is $k$-good, then we say that the index $k$ is good in step $t$ in the original process. 
(Recall that during step $t$, the configuration $\vx(t)$ given at the beginning of this step
is used to decide the next configuration $\X{t+1}$.)
This definition applies also to the mirrored process.
Note that if $\vol{S} \le m$ and the set $S_{+1}$ is empty, that is, 
no opinion is greater than $k$, {then this index} is (trivially) good.
If it is understood from the context whether we are referring to the original or mirrored process,
we will write simply 'a $k$-good configuration' and 'a good index'.
\end{definition}

A $k$-good configuration means that,
if we look at all edges with opinion at least $k+1$ at one end 
and opinion at most $k$ at the other end
(and there are ${\abs{E(S_{+1}, V\setminus S_{+1})}} \ge \Phi(G)\cdot\vol{S_{+1}}$ such edges by definition of conductance),
a significant number of them (at least $\tfrac{1}{4}\Phi(G)\cdot\vol{S_{+1}}$) 
bypass opinion $k$ (that is, they join opinions at least $k+1$ with opinions at most $k-1$).
We will show that this implies 
that the potential corresponding to such index $k$ decreases, in expectation,
at least by an amount that is proportional to the conductance.

In \Autoref{sec:goodindex} we prove that either in the original or in the mirrored process many of the indices are  frequently good.
More precisely, using a pigeonhole argument together with structural properties of the configuration, we show that for each period of computation, there exists an index
$k^\star \in [K/2+1, (\nicefrac{3}{4})K]$
that is good in a constant fraction {of} all steps in this period, 
either in the original or in the mirrored process.

Finally, in \Autoref{sec:lemmaproof}, we combine the results about the evolution of the potential 
and the existence of good indices.
Suppose an index $k^{\star}$ is good in the original process in a constant fraction of steps 
over a sufficiently long period.
Since the expected value of 
the potential $\Psi^{(+)}_{k^\star}(\Vec{X}(t))$ 
decreases (at least by some specific positive quantity) 
whenever index $k^\star$ is good and never increases,
then \whp it reaches zero by the end of this period.
This implies that the discrepancy drops {to at most} $k^\star \le (\nicefrac{3}{4})K$, completing the proof of the proposition.

\subsection{Potential Drop in One Step 
(For a Fixed \texorpdfstring{$k$)}{k}}
\label{sec:potentialdrop}

In this section, we analyze the expected drop of potential.
To this end, we fix an arbitrary integer index 
{$k\in[1,K-1]$} and show two properties of the 
positive potential function w.r.t. $k$.
First, this potential is a supermartingale, so non-increasing in expectation.
Second, if the configuration is $k$-good at the beginning of the current step, then,
in expectation, the potential decreases in this step
by a multiplicative factor which depends on the graph conductance.
Formally, in this section we prove the following lemma.

\newcommand{\Ext}[1]{\underset{_{\Vec{X}(t) = \Vec{x}}}{\mathbf{E}}\left[#1\right]}

\begin{lemma}
\label{lemma:potentialdrop}
Let $G \coloneqq (V,E)$ be a  graph with conductance $\Phi(G)$ and  $\vec{x}$  an arbitrary configuration with opinions in $[0,K]$, where $K\ge 3$.
{For any integer $k\in[1,K-1]$, the following holds.}
\begin{enumerate}[itemsep=2pt,topsep=2pt,parsep=0pt,partopsep=0pt,leftmargin=14pt]
\item \label{lem:nonincreasing-potential}
The positive potential is a supermartingale as it is non-increasing in expectation:
\begin{equation}\label{eq:nonincreasingpotential}
    \E{\Psi^{(+)}_k(\Vec{X}(t+1))\mid \X{t} = \Vec{x}}
    \;\le\;
    \Psi^{(+)}_k(\Vec{x})
\end{equation}

\item \label{lem:decreasing-potential}
If configuration $\Vec{X}(t) = \Vec{x}$ is $k$-good in step $t$, 
then the positive potential exhibits multiplicative drift:
\begin{equation}\label{eq:decreasingpotential}
    \E{\Psi^{(+)}_k(\Vec{X}(t+1)) \mid \X{t} = \Vec{x}\: \text{and $k$-good for this step}}
    \;\le\;
    \left(1 - \frac{\Phi(G)^2}{5(K-k)n}\right)\cdot \Psi^{(+)}_k(\Vec{x})
\end{equation}
\end{enumerate}
The same holds for the positive potential $\Psi^{(+)}_k(\Vec{X'}(t+1))$ for  the mirrored process.
\end{lemma}
We begin the analysis by introducing some notation and assumptions that we use in this section.
Throughout the whole proof, we condition all expectations and probabilities on the event that $\{\Vec{X}(t) = \Vec{x} \}$ for some valid configuration $\Vec{x} \in [K]^n$.
For readability, we just write $\E{\cdot}$ and $\pr{\cdot}$ instead of $\E{\cdot \mid \Vec{X}(t) = \Vec{x}}$ and $\pr{\cdot \mid \Vec{X}(t) = \Vec{x}}$, respectively.
Further, since we only consider a fixed $k$ and the positive potential {w.r.t.}\ $k$, 
we write $f(x)$ instead of $\varphi_k^{(+)}(x)$.
Thus $f(x)= \text{max}\{x-k, 0\}$.

We define the following random variable which measures the change of the potential in step $t$.
\begin{align*}
    \Delta(\Vec{X}(t+1)) &=
    \Psi^{(+)}_k(\Vec{X}(t+1)) - {\Psi^{(+)}_k(\Vec{X}(t))}
    =\sum_{i \in V} d_i \cdot f(X_i(t+1))^2
    -
    \sum_{i \in V} d_i \cdot f(X_i(t))^2,
\end{align*}
Since we condition on $\Vec{X}(t) = \Vec{x}$, we have
\[
\E{\Delta(\Vec{X}(t+1))} =
\E{\sum_{i \in V} d_i \cdot f(X_i(t+1))^2}
-
\sum_{i \in V} d_i \cdot f(x_i)^2.
\]
Given this definition, it suffices to show that $\E{\Delta(\Vec{X}(t+1))} \leq 0$ to prove part (1) of the lemma.
To prove part (2), we need to show that if $\Vec{x}$ is $k$-good in step~$t$, then
\begin{equation}\label{eq:driftLem4.4}
\E{\Delta(\Vec{X}(t+1))}
\leq
-
\frac{\Phi(G)^2}{5(K-k)n}\cdot
\sum_{i \in V} d_i \cdot f(x_i)^2.
\end{equation}
Although $G = (V,E)$ is a simple \emph{undirected} graph,
it will be convenient to view $G$ as a bi-directed graph with directed edges $(i,j)$ and $(j,i)$ for each undirected edge $\{i,j\}$.
To avoid confusion, we use the notation $\{i,j\}$ and $E$ when referring to undirected edges, 
and $(i,j)$ and $\Vec{E}$ when referring to directed edges.
We also use the notation introduced in~\eqref{defn-Splus}.

Having established the preliminaries, we next show a technical claim that will 
immediately imply part (\ref{lem:nonincreasing-potential}) of the lemma 
and will contribute to the proof of part (\ref{lem:decreasing-potential}).
This claim gives a bound on the expected change of the potential 
in terms of the differences $f(x_i) - f(x_j)$ along the edges plus 
'error terms' $\mathds{1}_{\{|x_i - x_j| =1\}}$
for the edges where the opinion difference is exactly $1$.
This bound
will make the remainder of the proof of the lemma
simpler as it eliminates both the degree weights and the squared terms.

\begin{claim}\label{Claim1}
Let 
$S_{+1} = \{ i \in V \mid x_i \ge k+1 \}$ and
$E_{S_{+1}} = \{ \{i,j\} \in E: \text{$i$ or $j$, or both, in $S_{+1}$}\}$.  Then it holds:
\begin{align}
     \E{\Delta(\Vec{X}(t+1))}
     \le
     -\,
     \frac{1}{n}
     \sum_{\{i,j\} \in E_{S_{+1}}}
     \left(
     \abs*{f(x_i) - f(x_j)}
     -
     \mathds{1}_{\{|x_i-x_j |= 1\}}
     \right)
\label{EsumClaim1}
\end{align}
\end{claim}
\begin{proof}[Proof Sketch]
(The full proof is given in \Autoref{sec:proofClaim1}).
This claim follows by elementary calculations based on the following idea.
For two adjacent nodes $i$ and $j$ with $i\in S_{+1}$ and $x_i > x_j$ (so $f(x_i) - f(x_j)>0$), 
we can amortize the increase of the potential when the edge $(j,i)$ is chosen 
(that is, node $j$ is chosen and then selects its neighbor~$i$) with the decrease 
of the potential when the opposite edge $(i,j)$ is chosen.
Both these changes cancel out in expectation if $x_i-x_j=1$.
If $x_i-x_j\ge2$, then
the convexity of the square function implies that the expected decrease of the term $d_if(x_i)^2$ 
in the potential dominates the expected increase of the term $d_jf(x_j)^2$.
The difference between them is precisely captured (lower bounded) by the term $f(x_i) - f(x_j)$, and summing over all relevant edges yields~\eqref{EsumClaim1}.
\end{proof}

\paragraph{Proof of \autoref{lemma:potentialdrop}-(\ref{lem:nonincreasing-potential})} 
We use~\eqref{EsumClaim1} to prove the first statement of the lemma.
Each summand in~\eqref{EsumClaim1} is non-negative:
since $\max{}\{x_i,x_j\}\ge k+1$, 
if $|x_i-x_j|\ge 1$ then $\abs*{f(x_i)-f(x_j)}\ge 1$ and the summand is non-negative,
and if $x_i = x_j$, then the summand is equal to $0$.
Thus $\E{\Delta(\Vec{X}(t+1))} \le 0$, proving that the potential is non-increasing in expectation and therefore a supermartingale.

\paragraph{Proof of \autoref{lemma:potentialdrop}-(\ref{lem:decreasing-potential})} 
To prove the second part of the lemma, we require some more techniques.
First, we show that the edge-wise differences $f(x_i) - f(x_j)$ can be bounded in terms of the graph's conductance.
Using techniques from the analysis of random walks and continuous load balancing, we prove the following claim.
\begin{claim}
\label{claim:averaging}
Suppose that $\vol{S} \leq m$, then
\begin{align}
    \sum_{\{i,j\} \in E_{S_{+1}}}
    \abs*{
    f(x_i) - f(x_j)
    }
    \ge
    \Phi(G)
    \cdot
    \sum_{i \in V} d_i f(x_i).
\label{eq:edgediff}
\end{align}
\end{claim}
\begin{proof}
We show this inequality using
a combinatorial technique from {Mihail}~\cite{Mihail89},
where it was used to bound the convergence of simple random walks on $d$-regular graphs.
We slightly adapt the analysis from \cite{Mihail89} to fit our process.
Assume w.l.o.g. the nodes are ordered by decreasing opinion, i.e., $x_1 \geq x_2 \geq \ldots \geq x_n$. This way $i < j$ implies $x_i \geq x_j$ and $f(x_i) \geq f(x_j)$.
Thus,
\begin{align}
    \sum_{\{i,j\} \in E_{S_{+1}}}
    \abs*{
    f(x_i) - f(x_j)
    } 
    = \sum_{\substack{(i,j) \in \Vec E \\ i \le j_k,\: i<j }} \left(f(x_i)-f(x_j)\right)
    = \sum_{\substack{(i,j) \in \Vec E \\ i<j }} \left(f(x_i)-f(x_j)\right).
\label{eq:klkl}
\end{align}
For each $1\le \ell\le n$,
let
$A_\ell:= \{1, 2, \ldots, \ell\}$
be the set of the first $\ell$ nodes after ordering by decreasing opinion.
The set of edges leading out of $A_\ell$ is
    $\mathbb{O}(A_\ell) \coloneqq \left\{ \{v,w\} \in {E} \mid v \in A_\ell\text{ and } w \not\in A_\ell \right\}$.
Inequality~\eqref{eq:edgediff} will follow 
from~\eqref{eq:klkl} and 
the following identity shown 
in {Mihail}~\cite{Mihail89}, 
which holds for any 
graph $G(V,E)$ with $V=\{1,2,\ldots,n\}$ 
and any function $g:V\rightarrow {\mathbb{R}}$.
\begin{align}
\sum_{\substack{(i,j) \in \Vec E \\ i<j }} (g(i)-g(j))
&= \sum_{\substack{(i,j) \in \Vec E \\ i<j }} \, \sum_{\ell=i}^{j-1}(g(\ell)-g({\ell+1}))
=\sum_{\ell=1}^{n-1} \sum_{\substack{(i,j) \in \Vec E \\ i\le \ell <j }} (g(\ell)-g({\ell+1}))\nonumber\\
&=\sum_{\ell=1}^{n-1}(g({\ell})-g({\ell+1}))|\mathbb{O}(A_\ell)|.
\label{eq:ppp}
\end{align}
Let $j_k$ be the last index with $x_{j_k} > k$, that is, $S_{+1} = \{1,2,\ldots,j_k\}$,
and observe that $f(x_i) = 0$ for all $i > j_k$.
Using~\eqref{eq:klkl} 
and~\eqref{eq:ppp} with $g(i) = f(x_i)$,
we obtain~\eqref{eq:edgediff}:
\begin{align}
    \sum_{\{i,j\} \in E_{S_{+1}}} &
    \abs*{
    f(x_i) - f(x_j)
    } 
    = \sum_{\substack{(i,j) \in \Vec E \\ i<j }} \left(f(x_i)-f(x_j)\right)
= \sum_{\ell=1}^{n-1}(f(x_{\ell})-f(x_{\ell+1}))|\mathbb{O}(A_\ell)|
= \sum_{\ell=1}^{j_k}(f(x_{\ell})-f(x_{\ell+1}))|\mathbb{O}(A_\ell)| \nonumber \\
&\ge \Phi(G) \sum_{\ell=1}^{j_k}(f(x_{\ell})-f(x_{\ell+1}))\vol{A_\ell}
= \Phi(G) \left( \sum_{\ell=1}^{j_k}f(x_{\ell})\cdot\vol{A_\ell} - \sum_{\ell=1}^{j_k} f(x_{\ell+1})\cdot\vol{A_{\ell}} \right) \label{eq:jkkm6a1}\\
&= \Phi(G) \left( \sum_{\ell=1}^{j_k}f(x_{\ell})\cdot\vol{A_\ell} - \sum_{\ell=1}^{j_k} f(x_{\ell})\cdot\vol{A_{\ell-1}} \right)\label{eq:jkkm6a2}\\
& = \Phi(G) \sum_{\ell=1}^{j_k} f(x_{\ell})\cdot d_\ell \; = \; \Phi(G) \sum_{\ell=1}^n f(x_{\ell})\cdot d_\ell.  
\label{eq:jkkm6a3}
\end{align}
The inequality on line~\eqref{eq:jkkm6a1} 
holds because
by the definition of the conductance $\Phi(G)$ given in~\eqref{eq:conductance}, we have $|\mathbb{O}(A_\ell)| \geq \Phi(G)\cdot \vol{A_\ell}$, for any $\ell\ge 1$ 
such that $\vol{A_\ell} \leq m$,
and by the definition of $j_k$ and the assumptions of the claim, for $1 \le\ell\le j_k$,
$\vol{A_\ell} \le \vol{A_{j_k}} = \vol{S_{+1}} \le \vol{S} \le m$.
The equality on line~\eqref{eq:jkkm6a2}
holds because $\vol{A_0}=\vol{\emptyset}=0$ and
$f(x_{j_k+1})=0$.
The first equality on line~\eqref{eq:jkkm6a3}
holds because 
for each $1\le \ell\le n$,
$\vol{A_\ell} = \sum_{i=1}^\ell d_i$, so $\vol{A_\ell}-\vol{A_{\ell-1}}=d_\ell$.
The second equality  on line~\eqref{eq:jkkm6a3}
holds since $f(x_\ell)= 0$ for $\ell > j_k$.

\end{proof}

The next step is to use the inequalities~\eqref{EsumClaim1} from \autoref{Claim1} 
and~\eqref{eq:edgediff} from \autoref{claim:averaging}
to prove the following claim.

\begin{claim}
\label{Claim2}
Suppose that the configuration $\Vec{x}$ in step $t$ is $k$-good.
Then,
\begin{align}
\E{\Delta({\Vec{X}(t+1)})}
\le
-
\frac{\Phi(G)^2}{5n}\cdot
\sum_{i \in V} d_i f(x_i).
\label{Claim2Value}
\end{align}
\end{claim}

\begin{proof}
From~\eqref{EsumClaim1} and~\eqref{eq:edgediff}, we have,
\begin{align} 
     \E{\Delta(\Vec{X}(t+1))}
     \le
     -
     \frac{1}{n}\Phi(G)\cdot
     \sum_{i \in V} d_i f(x_i)
     +
     \frac 1n
     \sum_{\{i,j\} \in E_{S_{+1}}}
     \!\!\!\! \mathds{1}_{\{|x_i-x_j |= 1\}}.
\label{Part2Ineq1}
\end{align}

To show that \eqref{Claim2Value} holds when the configuration $\Vec{x}$ in step $t$ is $k$-good,
we distinguish two cases which are separated 
by setting a threshold for the value $\sum_{i \in V} d_i f(x_i) / \vol{S_{+1}}$.

\medskip
\noindent
\textbf{Case 1:}
Suppose $\sum_{i \in V} d_i f(x_i)\ge \frac{5}{4\cdot{\Phi(G)}}\cdot{\vol{S_{+1}}}$.
In this case
the error term sum in \eqref{Part2Ineq1}, 
which is maximized when every edge with an endpoint in $S_{+1}$ has difference exactly $1$, 
is at most
$\sum_{i\in S_{+1}} d_i = \vol{S_{+1}} \le \frac{4}{5}\Phi(G)\cdot\sum_{i \in V} d_i f(x_i)$.
Using this in~\eqref{Part2Ineq1} gives
\begin{align*}
\E{\Delta(\Vec{X}(t+1))}
&\le - \frac{1}{n}\Phi(G)\cdot \sum_{i \in V} d_i f(x_i)
+ \frac{4}{5n}\Phi(G)\cdot\sum_{i \in V} d_i f(x_i)
= -\,\frac{1}{5n}{\Phi(G)}\cdot
\sum_{i \in V} d_i f(x_i)\\
&\le
-\,\frac{1}{5n}\Phi(G)^2\cdot\sum_{i \in V} d_i f(x_i).
\end{align*}
The last inequality above holds because $\Phi(G) \le 1$.
Note that in this case, when $\sum_{i \in V} d_i f(x_i)$ is relatively large, we do not use the assumption that the configuration is $k$-good in this step. 

\noindent
\textbf{Case 2:}
Suppose the opposite, that
$ \sum_{i \in V} d_i f(x_i) < \frac{5}{4\cdot{\Phi(G)}}\cdot{\vol{S_{+1}}}.
$
In this case, we refer back directly to \eqref{EsumClaim1} instead of \eqref{Part2Ineq1}.
First, recall that all summands in the sum~\eqref{EsumClaim1} are non-negative,
as shown in the proof of part (1) of the lemma.
In the following, we only count the contribution to this sum
from the edges
in $E(S_{+1},\ovl S)$, that is, from the edges $(i,j)$ with $x_i\ge k+1$ and $x_j\le k-1$.
As each such edge contributes at least~$1$ to this sum
(since $f(x_i) - f(x_j) \ge 1$ and $\mathds{1}_{\{|x_i-x_j |= 1\}}=0$), we have
\begin{align}
\E{\Delta(\Vec{X}(t+1))}
\le -\,\frac{1}{n}\cdot \abs{E(S_{+1},\ovl S)}.
\label{yalwnn88}
\end{align}

Second, since $\Vec{x}$ is $k$-good by assumption, we have by \autoref{def:kgood} that
\begin{align}
\abs{E(S_{+1},\ovl S)} \geq \tfrac{1}{4}\Phi(G)\cdot\vol{S_{+1}}.
\label{eq:sizeboundZ}
\end{align}
From~\eqref{yalwnn88}, \eqref{eq:sizeboundZ} and the assumption of this case, we obtain~\eqref{Claim2Value}:
\[
\E{\Delta(\Vec{X}(t+1))}
\le
-\,\frac{1}{4n}\Phi(G)\cdot \vol{S_{+1}}
\le -\, \frac{1}{4n}\Phi(G)\cdot \frac{4}{5}\Phi(G)\sum_{i \in V} d_i f(x_i)
= -\, \frac{1}{5n}\Phi(G)^2\sum_{i \in V} d_i f(x_i).
\]
\end{proof}

To conclude the proof of \autoref{lemma:potentialdrop}-(2), 
observe that~\eqref{eq:driftLem4.4} follows from~\eqref{Claim2Value} and
the fact that for each $i\in V$, $f(x_i) \le K-k$.

\subsection{Existence of Good Indices}
\label{sec:goodindex}

In this section, we show that there exists an index 
$k^{\star}\in \left[K/2+1,(\nicefrac{3}{4})K\right]$ 
which is frequently good in either the original or the mirrored process,
so consequently, by \autoref{lemma:potentialdrop}-(2), 
the corresponding (positive) potential frequently decreases.
\begin{lemma}
\label{lemma:goodindexexists}
{Let $K\ge\tfrac{32\log n}{\Phi(G)}$ be a fixed integer range bound, 
with initial opinions in $[0,K]$, and consider a sequence of\/~$T$ steps $0,1,\ldots,T-1$.}
Then there exists an index
$k^{\star} \in
\left[K/2+1,(\nicefrac{3}{4})K\right]$
that is good in at least $T/4$ steps in the original process or 
in at least $T/4$ steps in the mirrored process.
\end{lemma}

\begin{proof}
For each step $t$, define
\[
S_k(t) \coloneqq \{i \in V : X_i(t) \ge k\},
\qquad
\overline S_k(t) \coloneqq V \setminus S_k(t).
\]
Since $\vol{S_{K/2+1}(t)} + \vol{\overline S_{K/2+1}(t)} = 2m$, it follows that in every step $t$,
at least one of the two sets
${S_{K/2+1}(t)}$ and ${\overline S_{K/2+1}(t)}$ has volume {at most $m$}.
Therefore, there exists a subset $\mathcal{L} \subseteq [T-1]$ of at least  $\lceil T/2\rceil$ steps
such that either
\[
\vol{S_{K/2+1}(t)} \leq m,
\quad \text{for all } t \in\mathcal{L} ,
\qquad \text{ or } \qquad
\vol{\ovl S_{K/2+1}(t)} \leq m,
\quad \text{for all } t \in\mathcal{L}.
\]
Assuming  that the former holds, we show that
there is an index which is good in at least half of
the steps in $\mathcal{L} $ in the original process.
Otherwise, in an analogous way, we can show   there is an index which is good in at least half of the
steps in $\mathcal{L}$ in the mirrored process.
In the remainder of the proof, we refer to the original process and
restrict attention to steps $t\in \mathcal{L}$.
Recall from  \eqref{goodk}, that an index $k \ge K/2+1$ is good  in step $t$ if
\begin{align*}
    \abs{E(S_{k+1}(t),\ovl S_k(t))} \geq \tfrac{1}{4}\Phi(G)\cdot\vol{S_{k+1}(t)}.
\end{align*}
This implies that 
for a \emph{bad} (not good) index $k$ in step $t$, it holds
\begin{equation}\label{eq:badindex}
 m \ge  \vol{S_k(t)} > \left(1+\tfrac{3}{4}\Phi(G)\right)\cdot\vol{S_{k+1}(t)}.
\end{equation}
To see this, consider the difference
$\vol{S_k(t)} - \vol{S_{k+1}(t)}$, 
\begin{align}
\vol{S_k(t)} - \vol{S_{k+1}(t)}
&= {\textstyle\sum}\left\{ d_i:\; {i\in S_k(t) \setminus S_{k+1}(t)} \right\} \nonumber\\
&\ge \abs{E(S_{k+1}(t),\, S_{k}(t)\setminus S_{k+1}(t))}
\;=\; \abs{E(S_{k+1}(t),\, \ovl S_{k+1}(t))}
- \abs{E(S_{k+1}(t),\, \ovl S_{k}(t))} \nonumber\\
&> \Phi(G)\cdot\vol{S_{k+1}(t)} -
\tfrac{1}{4}\Phi(G)\cdot\vol{S_{k+1}(t)}
\; =\; \tfrac{3}{4}\Phi(G)\cdot\vol{S_{k+1}(t)}.\label{llks231h}
\end{align}
Inequality~\eqref{llks231h} implies Inequality~\eqref{eq:badindex}.

Using~\eqref{eq:badindex}, we first show that in any step $t$, the number of bad indices in
$\left[K/2+1,K-1\right]$ is small.
Fix a step $t$ and let $k_1 > k_2 > \cdots > k_\ell$ be the bad indices in 
$\left[K/2+1,K-1\right]$ in this step.
Noting that, whenever $q\ge r$, $S_q(t) \subseteq S_r(t)$, and hence $\vol{S_q(t)} \le \vol{S_r(t)}$,
it follows from~\eqref{eq:badindex} that
\begin{align*}
m \ge \vol{S_{K/2+1}(t)} \ge \vol{S_{k_{\ell}}(t)}
&> \left(1+\tfrac{3}{4}\Phi(G)\right)\cdot \vol{S_{k_{\ell}+1}(t)}
\geq \left(1+\tfrac{3}{4}\Phi(G)\right)\cdot\vol{S_{k_{\ell-1}}(t)}.
\end{align*}
By iterating the above for indices $k_\ell, k_{\ell-1}, \ldots, k_{1}$, we obtain
\begin{equation}\label{Eq:badindices}
m \ge \vol{S_{k_\ell}(t)}
>
\left(1+\tfrac{3}{4}\Phi(G)\right)^\ell\cdot \vol{S_{k_1+1}(t)}
\ge
\left(1+\tfrac{3}{4}\Phi(G)\right)^\ell,
\end{equation}
where the last inequality holds because $\vol{S_{k_1+1}(t)} \ge 1$
(the index $k_1$ is bad, so the set $S_{k_1+1}(t)$ cannot be empty).

As $m \le n^2$, Inequality~\eqref{Eq:badindices} implies
\[
\ell < \frac{\log(n^2)}{\log (1+(\nicefrac{3}{4})\Phi(G))}
\le \tfrac{4\log n}{\Phi(G)},
\]
where the last inequality holds because $\log(1+x)\ge (\nicefrac{2}{3})x$ for $0<x\le 1$.

As $K\ge{32\log n}/{\Phi(G)}$, 
the number of bad indices in $\left[K/2+1,K-1\right]$
is at most
$\frac{4\log n}{\Phi(G)} \le K/8$.
This implies that at most half of the indices in 
$\left[K/2+1,(\nicefrac{3}{4})K\right]$ are bad.
Thus, in any step $t\in\mathcal{L}$,
at least half of the indices in $\left[K/2+1,(\nicefrac{3}{4})K\right]$ are good, 
so there must be an index in $\left[K/2+1,(\nicefrac{3}{4})K\right]$ which is good in at least half
of the steps in~$\mathcal{L}$.
Such an index $k^\star$ is good in at least $|\mathcal{L}|/2 \ge T/4$ steps in $[T-1]$.
\end{proof}

\subsection{Putting  Everything Together: Completion of the proof of \autoref{lemma:phasescombined}}
\label{sec:lemmaproof}

First, using the bound 
on the expected decrease of the potentials for good indices 
given in \autoref{lemma:potentialdrop},
we show in \autoref{lemma:goodindexdrops} that for any index $k$ 
and a sufficiently long period of computation,
the probability that the potential w.r.t.\ this index is still positive 
at the end of this period
while the index has been good in a significant number of steps during this period is negligible.
Then we complete the proof of \autoref{lemma:phasescombined} using
Lemmas \ref{lemma:goodindexexists} and~\ref{lemma:goodindexdrops}.

\begin{lemma}
\label{lemma:goodindexdrops}
Consider an \ProtAV process on a connected graph $G$, starting from a configuration whose opinions lie in the fixed range $[0,K]$ , where $K \ge 3$.

Let $k\in[1,K-1]$ be a fixed integer.
Let for $t\ge 0$, $\tau_k(t) \coloneqq |\{\tau\in[t]: \text{$k$ is good at step $\tau$}\}|$~ be the number
of steps in $[t]=\{0,\ldots,t\}$ when index $k$ is good for the original 
process, and set $\tau_k(-1)=0$. 

 Let  {$T$ be an integer satisfying $T\ge\left\lceil20(c+4)(K-k){n\log((K-k)n)}/{\Phi^2(G)}\right\rceil$}, where $c>0$ is a constant.
Let $\mathcal{B}_k(T)$
be the event that index $k$ is good in at least
quarter of the steps in $[T-1]$ but the potential w.r.t.\ this index is still positive at the 
beginning of step $T$:
\begin{align}
    \mathcal{B}_k(T) \coloneqq \big\{\tau_k(T-1) \ge T/4\big\} 
    \; \cap \; \left\{\Psi^{(+)}_k(\X{T}) \ge 1 \right\}.
\label{dfn:BkT}
\end{align}
Then
\[
\pr{\mathcal{B}_k(T)} \; \le \; \frac{1}{(K-k)^{c+2}n^{c+2}}.
\]
The same holds for the analogous event $\mathcal{B}'_k(T)$ in the mirrored process.
\end{lemma}

\begin{proof}
We begin with some definitions and assumptions.
{Let $\mathcal{F}(t)=\mathcal{F}_t=\sigma(\omega_0,\ldots,\omega_{t-1})$ be the filtration of interaction choices made before step $t$, with $\mathcal{F}_0$ trivial for a deterministic initial state. The configuration $\X{t}$ is measurable with respect to this filtration, but configuration history need not determine all past interaction choices.}

We define the following auxiliary random process:
\begin{align}\label{martingale-psi}
    \psi(t) \coloneqq
    \begin{cases}
        \Psi^{(+)}_k(\X{0}), & \text{if $t=0$},\\[2pt]
        \Psi^{(+)}_k(\X{t})\cdot\left(1- \dfrac{\Phi(G)^2}{5(K-k)n}\right)^{-\tau_k(t-1)}, & \text{if $t>0$}.
    \end{cases}
\end{align}
Note that $\psi(t)$ is measurable with respect to the filtration $\mathcal{F}(t)$,
because $\tau_k(t-1)$ and $\X{t}$ are determined by $\mathcal{F}(t)$.
Since $\tau_k(t)$ is also determined by $\mathcal{F}(t)$,
then $\E{\psi(t+1) \mid \mathcal{F}(t)}$ is the expectation of $\psi(t+1)$ given $\X{t}$
and $\tau_k(t)$,
and we can show, using \autoref{lemma:potentialdrop}, 
that $(\psi(t))_{t \geq 0}$ is a supermartingale, i.e., for all $t\ge 0$,
we have $\E{\psi(t+1) \mid \mathcal{F}(t)} \le \psi(t)$.
Indeed, for $t\ge 0$, we have
\begin{align}
\E{\psi(t+1) \mid \mathcal{F}(t)} 
&= \left(1- \dfrac{\Phi(G)^2}{5(K-k)n}\right)^{-\tau_k(t)} \cdot \E{\Psi^{(+)}_k(\X{t+1})\mid \mathcal{F}(t)}
\nonumber\\
&\le \left(1 - \frac{\Phi(G)^2}{5(K-k)n}\right)^{-\tau_k(t)} 
\cdot \left(1 - \frac{\Phi(G)^2}{5(K-k)n}\right)^{\tau_k(t)-\tau_k(t-1)}\Psi^{(+)}_k(\X{t}) 
\;\:=\;\: \psi(t). \label{eq:supemartingalepsi}
\end{align}
The inequality~\eqref{eq:supemartingalepsi} follows from \autoref{lemma:potentialdrop}
because $\tau_k(t) = \tau_k(t-1)+1$, if index $k$ is good in step $t$ 
(and apply \eqref{eq:decreasingpotential} from \autoref{lemma:potentialdrop}), 
and $\tau_k(t) = \tau_k(t-1)$ otherwise (and apply \eqref{eq:nonincreasingpotential} 
from \autoref{lemma:potentialdrop}).

Now define the stopping time
\begin{equation}\label{T-time}
T'
\coloneqq
\min\left\{
T,\;
\inf\{t\ge 0:\tau_k(t-1)\ge T/4\}
\right\}.
\end{equation}
As $(\psi(t))_{t\ge0}$ is a nonnegative supermartingale and the stopping time $T'$ is bounded by $T$,
the conditions of Doob's optional stopping theorem (cf.\ \autoref{doob}) are satisfied, so we have
\begin{align}
\E{\psi(T')}
\le
\psi(0).\label{supmart1}
\end{align}
Once the potential $\Psi_k^{(+)}$ reaches zero, 
all opinions are at most $k$, and this property is preserved by every subsequent update. 
Thus, on $\mathcal{B}_k(T)$, positivity of $\Psi_k^{(+)}$ at time $T$ and $T'\le T$ imply $\Psi_k^{(+)}(\X{T'})\ge1$. 
Further, since $\mathcal{B}_k(T)$ implies $\tau_k(T-1)\ge T/4$, 
by the definition of $T'$ in \eqref{T-time} we have $\tau_k(T'-1)\ge T/4$.
Hence, if $\mathcal{B}_k(T)$ occurs, we have
\begin{align*}
\psi(T')
&\;=\;
\Psi_k^{(+)}(\X{T'})
\cdot
\left(1-\frac{\Phi(G)^2}{5(K-k)n}\right)^{-\tau_k(T'-1)}
\;\ge\;
\left(1-\frac{\Phi(G)^2}{5(K-k)n}\right)^{-T/4}.
\end{align*}
Therefore,
\begin{align}
\E{\psi(T')}
&\;\ge\;
\E{\psi(T')\,|\, \mathcal{B}_k(T)} \cdot \pr{\mathcal{B}_k(T)} 
\;\ge\;
\left(1-\frac{\Phi(G)^2}{5(K-k)n}\right)^{-T/4} \cdot \pr{\mathcal{B}_k(T)}.\label{supmart2}
\end{align}
{Since $\psi(0)=\Psi_k^{(+)}(\x{0})\le2m(K-k)^2\le n(n-1)(K-k)^2\le n^2(K-k)^2$,}
Inequalities~\eqref{supmart1} and~\eqref{supmart2} give
\begin{align*}
{(K-k)^2 n^2}
\ge
\left(1-\frac{\Phi(G)^2}{5(K-k)n}\right)^{-T/4}\cdot\pr{\mathcal{B}_k(T)}.
\end{align*}
Solving for $\pr{\mathcal{B}_k(T)}$ and using
{$T\ge20(c+4){(K-k)n\log((K-k)n)}/{\Phi(G)^2}$}, we obtain the required bound 
on $\pr{\mathcal{B}_k(T)}$:
\begin{align*}
\pr{\mathcal{B}_k(T)}
&\;{\le}\;
{(K-k)^2 n^2}
\left(1-\frac{\Phi(G)^2}{5(K-k)n}\right)^{T/4}\\
&\;{\le}\;
{(K-k)^2 n^2}
\left(1-\frac{\Phi(G)^2}{5(K-k)n}\right)^{\tfrac{5(K-k)n}{\Phi(G)^2} {(c+4)} \log((K-k)n)}\\
&\;\leq\; {(K-k)^2 n^2} e^{-{(c+4)}\log{((K-k)n)}} 
{\;=\;\frac{1}{(K-k)^{c+2}n^{c+2}}.
}
\end{align*}
\end{proof}

We are finally in the position to complete the proof of \autoref{lemma:phasescombined}.
{Set $T=T_K$, with $T_K$ as defined in the proposition. Since $K-k\le K$, 
this choice of $T$ satisfies the requirement of \autoref{lemma:goodindexdrops} 
for every index $k$ used below.}
Let $\mathcal{H}$ be the event that for each $k \in \left[K/2+1, (\nicefrac{3}{4})K\right]$
both potentials (original and mirrored) are still positive after {$T$} steps, that is,
\[ \mathcal{H} \coloneqq \bigcap_{k=K/2+1}^{(\nicefrac{3}{4})K} \left\{
\Psi_k^{(+)}(\X{T}) \ge 1 \; \mbox{ and }\; \Psi_{k}^{(+)}(\Vec{X'}({T})) \ge 1 \right\}.
\]
If $\mathcal{H}$ does not occur, 
then the discrepancy must have dropped to $(3/4)K$ by time $T$, as desired.
Thus it suffices to upper bound the probability of $\mathcal{H}$.
For each $k \in [K/2+1, (\nicefrac{3}{4})K]$, consider the event
$\mathcal{B}_k(T)$ defined in \autoref{lemma:goodindexdrops} in~{\eqref{dfn:BkT}} for the original process 
and the analogous event $\mathcal{B}'_k(T)$ for the mirrored process.
By \autoref{lemma:goodindexexists}, there must exist
an index $k^\star\in[K/2+1, (\nicefrac{3}{4})K]$
such that $\tau_k(T-1) \ge T/4$
in the original or mirrored process.
Hence, if $\mathcal{H}$ occurs, then for such index $k^\star$,
the event $\mathcal{B}_{k^\star}(T)$ or the event $\mathcal{B}'_{k^\star}(T)$ occurs, so
$\mathcal{H} \subseteq \bigcup_{k=K/2+1}^{(\nicefrac{3}{4})K} \left( \mathcal{B}_k(T) \cup \mathcal{B}'_k(T)\right)$,
and by the union bound and \autoref{lemma:goodindexdrops},
\begin{align*}
\pr{\mathcal{H}}\;
&\;\le\;
\pr{\bigcup_{k=K/2+1}^{(\nicefrac{3}{4})K} \left( \mathcal{B}_k(T) \cup \mathcal{B}'_k(T)\right)}
\;\le\;
\sum_{k=K/2+1}^{(\nicefrac{3}{4})K} \left(\pr{\mathcal{B}_k(T)} + \pr{\mathcal{B}'_k(T)}\right)\\
&\;\le\;
\frac{K+1}{4}\cdot\frac{2}{(K/4)^{c+2} n^{c+2}}
\;=\;
O\left(\frac{1}{K^{c+1} n^{c+2}}\right). \qedhere
\end{align*}

\section{Proof of Theorem~\ref{thm:main}, bound~(\ref{boundonexpectation})}
\label{Sec-Thm2.1-2nd-bound}
Our proof of the 'w.h.p.' bound~\eqref{boundwhp} on convergence time in \autoref{thm:main}
is completed. As outlined in \Autoref{sec:outlines-of-proofs}, that proof 
was split into two parts.
The first part dealt with
the discrepancy decreasing from the initial $K$ to $\alpha= (32 {\log n})/{\Phi}$, and
the detailed analysis of this is given in \Autoref{sec:prooftechnical}.
The second part considered the time needed to
decrease the discrepancy from $\alpha$ to $0$,
and the argument was based
on the reduction of DIV to iterative applications of 2-value pull voting
described in~\cite{CRS23}.
This gave immediately an $O\left(\alpha\,\mathcal{T}^{2V}_{G}\right)$ bound 
on the time of reducing the remaining discrepancy from $\alpha$ to zero, 
in expectation as well as with high probability,
where $\mathcal{T}^{2V}_{G}$ 
is the worst-case expected consensus time in 2-value pull voting. 
We now refine the analysis of the second part
to prove the 'in expectation' bound~\eqref{boundonexpectation} in \autoref{thm:main}.
For a quick reference, we repeat this bound below as a separate theorem.

\begin{theorem}
\label{thm:main-expectedtime}
\label{thm:main:version2}
Consider an \ProtAV process on a graph $G$ with conductance $\Phi(G)$ and the ratio of average to minimum
degree denoted by $\gamma(G)$. Assume
$\vec{x}$
is an arbitrary configuration with discrepancy~$K$.
Then
\begin{align}\label{bound:ExpectedTime1}
\E{T_G(\vec{x})}
\;=\;
O\!\left(
\frac{ K n \log(Kn)}{\Phi(G)^2}
+ \frac{\gamma(G) n^2}{\Phi(G)^2}\right).
\end{align}
\end{theorem}

For
graphs with constant conductance, this theorem implies a bound of $O(\gamma n^2)$
on the expected consensus time of \ProtAV  provided
 the initial discrepancy $K = O(n/\log n)$.
This $O(\gamma n^2)$ bound is of the same asymptotic order as the best known bound for
the expected consensus time in the 2-value pull voting for graphs with constant conductance,
as derived in~\cite{CooperR16}.
For expander graphs, 
where in addition to constant conductance we also have $\gamma = O(1)$,
this bound reduces further to
$O(n^2)$, matching the lower bound of $\Omega(n^2)$
for the expected consensus time of the 2-value pull voting
for this class of graphs.

Our analysis of the time needed to reduce the discrepancy from $\alpha$ to $0$
presented in this section
also refers to coupling the DIV process with 2-value pull voting, 
but it dives deeper into to the analysis of the convergence time of 2-value voting schemes
given by Cooper and Rivera~\cite[Section 4]{CooperR16}. 
At some point we deviate from that analysis 
to derive a more detailed bound on the convergence time
of 2-value pull voting,
which we state in \autoref{lem:2valpullvoteX} below and prove in \Autoref{sec:ProofVoting}. 
This more detailed bound is parameterized by
the initial imbalance of the support for the two opinions. 
For a set of nodes $S\subseteq V$,
let $\eta(S)\coloneqq \min\{\vol{S}, \voll(\ovl S)\}$, where $\ovl S = V\setminus S$.
For  2-value pull voting, assume that the opinions are $0$ and $1$
and let $S(t)$, $t\ge 0$, be the set of nodes with opinion $1$ at step~$t$. 

\begin{lemma}
\label{lem:2valpullvoteX}
With probability at least $\tfrac{1}{2}$, 
the convergence time of 2-value pull voting on a connected graph $G=(V,E)$ 
is $O\left(\tfrac{\eta(S(0)) \cdot n}{\Phi(G) \cdot \dmin}\right)$. 
\end{lemma}

For DIV and a step $\tau\ge 0$,
translate the current minimum opinion to $0$, denote 
the current discrepancy by $K$, so the current maximum opinion is $K$,
and assume that $K \ge 1$.
Let $S_0(\tau+t)$ and $S_K(\tau+t)$, for $t\ge 0$, be the sets of nodes with these 
two fixed opinions~$0$ and~$K$,
respectively, in step $\tau +t$.
We couple the process of evolution of sets $S_0(\tau+t)$ and $S_K(\tau+t)$ in DIV, for $t\ge 0$,
with the 2-value pull voting process $(S(t))_{t\ge 0}$,
where $S(0)$ is the set  with the smaller volume among $S_0(\tau)$ and $S_K(\tau)$.
Assume for discussion that $S(0) = S_0(\tau)$.
It is shown in~\cite[Lemma 6]{CRS23} that for this coupling, for all $t\ge 0$,
$S_0(\tau+t)\subseteq S(t)$ and
$S_K(\tau+t)\subseteq \ovl S(t)$.  
This means that for all $t\ge 0$,
\[ \underbrace{\eta(t) = \eta(S(t)) = \min\{\vol{S(t)}, \voll(\ovl S(t))\}}_{\text{2-value pull voting}} 
\; \ge \; 
\underbrace{\min\{\vol{S_0(\tau+t)}, \voll(S_K(\tau+t))\} = \tilde\eta(\tau+t)}_{\text{DIV}}.
\]

Therefore, if in 2-value pull voting the consensus is reached in step $t$,
meaning that $\eta(t) = 0$ and only one opinion remains in the system,
then in DIV we have $\tilde\eta(\tau+t) = 0$,
meaning that one of the fixed opinions $0$ or $K$ 
has vanished between the steps $\tau$ and $\tau +t$,
that is, the discrepancy has dropped (at least by 1) within $t$ steps after step $\tau$.
\autoref{lem:2valpullvoteX} gives therefore the following corollary.
(Note that both \autoref{lem:2valpullvoteX} and its \autoref{lem:dropbyone}
apply to any connected graph, not necessarily regular.)
\begin{corollary}
\label{lem:dropbyone}
In DIV, consider any time step $\tau \ge 0$ when the discrepancy of opinions 
is still at least $1$.
Then, with probability at least $\tfrac{1}{2}$, the discrepancy drops by $1$ within
the subsequent 
$O\left(\tfrac{ \tilde\eta(\tau) \cdot n}{\Phi(G) \cdot \dmin}\right)$ steps,
where $\tilde\eta(\tau)$ is the smaller of the support volumes
of the two extreme opinions at step $\tau$.
\end{corollary}

Informally, the bound on the consensus time in 2-value pull voting
given in \autoref{lem:2valpullvoteX},
which takes into account the initial size of the minority vote, 
says that the smaller the initial minority the faster consensus time.
\autoref{lem:dropbyone} translates this in the DIV context to a property that 
the smaller the size of an extreme opinion the faster it disappears.

We return now to the process of decreasing 
the discrepancy in DIV from $\alpha$ to $0$, 
and for brevity, for the remaining part of this section
we use $\Phi$, $\gamma$ and $\dminNG$, for $\Phi(G)$, $\gamma(G)$ and $\dmin$.
For $0<\epsilon\le 1$,
let $T(\epsilon)$ be the minimum $T$ such that,
starting from an arbitrary configuration with $\tilde\eta(0)\le \epsilon m$, 
the discrepancy drops by $1$ within $T$ steps with probability at least $1/2$.
From \autoref{lem:dropbyone}, 
$T(\epsilon)=O\left(\epsilon m n / (\Phi\, \dminNG)\right)
=O(\epsilon\gamma n^2/\Phi)$.
Our proof of \autoref{thm:main:version2} uses this bound on $T(\epsilon)$
and the following \Autoref{lem:smallK},
which can be viewed as an extension of \autoref{lemma:phasescombined} to smaller discrepancies.
(Recall that \autoref{lemma:phasescombined} covered the discrepancies which were at least
of order $(\log n) / \Phi$.)

\begin{proposition}\label{lem:smallK}
Consider the \ProtAV process on a simple connected graph with conductance $\Phi$
starting from an arbitrary configuration with discrepancy $K \ge 3$.
For $k \ge 1$, let
$\epsilon_k = e^{-\Phi k/16}$ and {$T_k=\left\lceil100{k n \log(kn)}/{\Phi^2}\right\rceil$}
(this setting of $T_k$ aligns with the definition of $T_K$ in \autoref{lemma:phasescombined} for $c=1$).
The worst-case expected time to reduce the discrepancy to $(\nicefrac{3}{4})K$ is
$O\left(T_K+ \sum_{k=(\nicefrac{3}{4})K}^K 
{T}(\epsilon_k)\right)$.
\end{proposition}

We use throughout this section the parameters $\epsilon_k$ and $T_k$ defined in this proposition.
We note that for initial discrepancies greater than $\tfrac{32 \log n}{\Phi}$,
this proposition is essentially
a version of \autoref{lemma:phasescombined}, since 
the sum $\sum {T}(\epsilon_k)$ in the bound disappears 
(for $\epsilon \le 1/n^2$, trivially $T(\epsilon) = 0$).
Thus, while \autoref{lem:smallK} and its proof 
work for any~$K\ge 3$, 
they provide new insight (beyond what we show in \autoref{lemma:phasescombined}) only for
initial discrepancies smaller than this threshold.

To prove \autoref{lem:smallK}, we will need two additional Lemmas~\ref{lemma:goodindexexists2} 
and~\ref{lemma:goodindexdrops2}.
\autoref{lemma:goodindexexists2} extends \autoref{lemma:goodindexexists} to smaller discrepancies, ascertaining
the existence of good indices also for small values of $K$, providing that the support for the extreme
opinions is sufficiently large.
We prove \autoref{lemma:goodindexexists2} by adapting the proof of \autoref{lemma:goodindexexists}.
\autoref{lemma:goodindexdrops2} is a slight generalization of \autoref{lemma:goodindexdrops},
allowing a longer observation period while keeping the required number of good steps fixed. 

While the DIV process is not completed yet,
let $\tilde\eta(t)$ track now the support for the \emph{current} two extreme opinions.
More precisely,
for $t\ge 0$, let 
\begin{align*}
\tilde\eta(t) &=\begin{cases}
\min\{\vol{S_{\mathsf{min}}(t)},\vol{S_{\mathsf{max}}(t)}\},&\text{if } \disc{\X{t}}>0,\\
0,&\text{if } \disc{\X{t}}=0,
\end{cases}
\end{align*} 
where $S_{\mathsf{min}}(t)$ and $S_{\mathsf{max}}(t)$
are the sets of nodes with the current minimum and maximum opinions, respectively, at step $t$.
Thus $\tilde\eta(t)$ refers to the same two extreme opinions while those opinions are still 
supported by some nodes, and is reset when one of those two opinions disappears. 

\begin{lemma}
\label{lemma:goodindexexists2}
Let $T \le \Tilde{T}$ and consider a sequence of 
steps $0,1,\ldots, \tilde{T}-1$
in the \ProtAV process starting from an arbitrary configuration with discrepancy $K\ge 3$.
If $\,\tilde\eta(t) > \epsilon_K m$ in at least $T$ of these steps, then there exists an index
$k^{\star} \in
\left[K/2+1,(\nicefrac{3}{4})K\right]$
that is good in at least $T/4$ steps in $[\tilde T -1]$ in the original process 
or in at least $T/4$ steps in $[\tilde T -1]$ in the mirrored process.
\end{lemma}

\begin{proof}
We use the notation $S_k(t) = \{i \in V : X_i(t) \ge k\}$, 
as in the proof of \autoref{lemma:goodindexexists}.
Let $\mathcal{D}\subseteq [\tilde{T}-1]$ be a set of $T$ steps such that for each $t\in \mathcal{D}$,
$\tilde\eta(t) > \epsilon_K m$.
There exists a subset $\mathcal{L} \subseteq \mathcal{D}$ of at least  $\lceil T/2\rceil$ steps
such that either $\vol{S_{K/2+1}(t)} \leq m$ for all $t \in\mathcal{L}$, or
$\vol{\ovl S_{K/2+1}(t)} \leq m$ for all $t \in\mathcal{L}$.
W.o.l.g.\ we assume the former and show that
there is an index which is good in the original process in at least half of
the steps in $\mathcal{L}$, so in at least $T/4$ steps in $[\tilde T -1]$.
Otherwise, we can show in an analogous way that there is an index which is good in the mirrored process
in at least half of the steps in $\mathcal{L}$.

We take any step $t\in\mathcal{L}$ and follow the proof of \autoref{lemma:goodindexexists}
up to the point where we establish that (cf. the inequalities~\eqref{Eq:badindices})
\begin{equation}\label{eqn:jk213}
m \ge \vol{S_{k_\ell}(t)}
>
\left(1+\tfrac{3}{4}\Phi\right)^\ell\cdot \vol{S_{k_1 +1}(t)},
\end{equation}
where
$k_1 > k_2 > \cdots > k_\ell$ are the bad indices in $\left[K/2+1,K-1\right]$
in this step.
At this point in the proof of \autoref{lemma:goodindexexists} we use the fact that 
$\vol{S_{k_1 +1}(t)}\ge 1$.
Here, we use instead the assumption that $\tilde\eta(t) > \epsilon_K m$.
This implies that $\vol{S_{k_1 +1}(t)}\ge \vol{S_K(t)}\ge \tilde\eta(t) > \epsilon_K m$.
Combining 
with~\eqref{eqn:jk213}, we get
\[
\left(1+\tfrac{3}{4}\Phi\right)^\ell\epsilon_K < 1,
\]
and
\[
\ell < \frac{\log(1/\epsilon_K)}{\log (1+(\nicefrac{3}{4})\Phi)} = \frac{\Phi K/16}{\log (1+(\nicefrac{3}{4})\Phi)}
\le \tfrac{K}{8}.
\]
Therefore, as in the proof of \autoref{lemma:goodindexexists}, 
in any step $t\in \mathcal{L}$, 
the number of bad indices in $\left[K/2+1,K-1\right]$
is at most
$K/8$, implying that at most half of the indices in 
$\left[K/2+1,(\nicefrac{3}{4})K\right]$ are bad.
Thus, in any step $t\in\mathcal{L}$,
at least half of the indices in $\left[K/2+1,(\nicefrac{3}{4})K\right]$ are good, 
so there must be an index in $\left[K/2+1,(\nicefrac{3}{4})K\right]$ which is good in at least half
of the steps in~$\mathcal{L}$.
Such an index $k^\star$ is good in at least $|\mathcal{L}|/2 \ge T/4$ steps in $[\tilde T-1]$.
\end{proof}

\begin{lemma}
\label{lemma:goodindexdrops2}
Let $K \ge 3$, fix $k\in[K/2+1,K-1]$, 
and let 
$T$ and $\tilde T$ be integers satisfying 
$$\tilde T\ge T\ge\left\lceil20(c+4)(K-k)n\log((K-k)n)/\Phi^2(G)\right\rceil,$$
where $c>0$ is a constant.
For the original process,
let $\tau_k(t) \coloneqq |\{\tau\in[t]: \text{$k$ is good in step $\tau$}\}|$, 
with $\tau_k(-1)=0$,
and define the following event that index $k$ is good in at least
$T/4$ steps in $[\tilde T-1]$ but the potential w.r.t.\ this index is still positive at the 
beginning of step $\tilde T$:
\begin{align}
    \mathcal{B}_k(\tilde T) \coloneqq \left\{\tau_k(\tilde T-1) \ge T/4\right\}  
    \;\cap\; \left\{\Psi^{(+)}_k(\Vec{X}(\tilde T)) \ge 1 \right\}.
\label{dfn:BkTextended}
\end{align}
Then
\[
\pr{\mathcal{B}_k(\tilde T)} \; \le \; \frac{1}{(K-k)^{c+2}n^{c+2}}.
\]
The same holds for the analogous event $\mathcal{B}'_k(\tilde T)$ in the mirrored process.
\end{lemma}

\begin{proof}
Use the nonnegative supermartingale $\psi(t)$ from the proof of \autoref{lemma:goodindexdrops}
(see~\eqref{martingale-psi}),
put $a=\Phi(G)^2/(5(K-k)n)$, and define the stopping time
\[
T'=\min\bigl\{\tilde T,\inf\{t\ge0:\tau_k(t-1)\ge T/4\}\bigr\}.
\]
On $\mathcal B_k(\tilde T)$, we have $\tau_k(T'-1)\ge T/4$,
and the potential at $T'$ is at least $1$ because zero potential is absorbing.
Thus $\psi(T')\ge(1-a)^{-T/4}$ on this event.
Since the stopping time $T'$ is bounded by $\tilde T$, 
the optional stopping theorem (\autoref{doob})
and $\psi(0)\le n^2(K-k)^2$ give (see the proof of \autoref{lemma:goodindexdrops} for details):
\begin{align*}
\pr{\mathcal B_k(\tilde T)}
\; \le \; n^2(K-k)^2(1-a)^{T/4}
\; \le \; n^2(K-k)^2e^{-(c+4)\log((K-k)n)}
\; = \; \frac{1}{(K-k)^{c+2}n^{c+2}}.
\end{align*}
The argument for the mirrored process is identical.
\end{proof}

\begin{proof}[Proof of \autoref{lem:smallK}]
Consider the \ProtAV process of decreasing the discrepancy one by one 
(in one step the discrepancy can reduce only by $1$) 
from $K$ down to $(\nicefrac{3}{4})K$.

For $0\le i \le K/4+1$, let $t_i$
be the first step when at the beginning of this step the discrepancy is $K-i$.
Thus $t_0 = 0$, $t_i < t_{i+1}$ for $0\le i \le K/4$,
and the discrepancy at (the beginning of) step $t_{K/4+1}$ is $K-(K/4+1)\le (\nicefrac{3}{4})K$.
We divide the computation into intervals 
$[t_i,t_{i+1}-1]$, $0\le i \le K/4$. That is,
the interval $[t_i,t_{i+1}-1]$ starts at the step when the discrepancy becomes $K-i$ and ends 
at the last step when it is still $K-i$.
We further subdivide each interval $[t_i,t_{i+1}-1]$ into subintervals
$[t_{i,j},t_{i,j+1}-1]$, $0\le j \le j_i-1$, in such a way that each subinterval is long enough 
to provide a fair chance to reduce the discrepancy by $1$.
The formal definition of these subintervals is as follows.

Take $i$ such that $0 \le i \le K/4$, define $t_{i,0} = t_i$
(the first step when the discrepancy is $K-i\,$), and assume inductively
that for some $j\ge 0$, $t_{i,j}$ has already been defined and that $t_{i,j}<t_{i+1}$
(that is, at step $t_{i,j}$, the discrepancy is still $K-i\,$).
Let 
\begin{align*}
    t'_{i,j} &=\min\{t\ge t_{i,j}: \tilde\eta(t) \le \epsilon_{K-i}\,m\}, \\
    t_{i,j+1} &=\min\{t'_{i,j}+T(\epsilon_{K-i}),\: t_{i+1}\}.
\end{align*}
If $t_{i,j+1} = t_{i+1}$, then $[t_{i,j},t_{i,j+1}-1]$ is the last subinterval of $[t_i,t_{i+1}-1]$
(as the discrepancy reduces from $K-i$ to $K-(i+1)$ in the next step $t_{i+1}$).
These subintervals of $[t_i,t_{i+1}-1]$ should be viewed in the following way.
Subinterval $[t_{i,j},t_{i,j+1}-1]$ starts with discrepancy $K-i$ at (the beginning of) step $t_{i,j}$,
reduces $\tilde\eta(t)$ to $\epsilon_{K-i}\,m$ by step $t'_{i,j}\ge t_{i,j}$
($t'_{i,j} = t_{i,j}$, if $\tilde\eta(t_{i,j}) \le\epsilon_{K-i}\,m$), and then 
gives the process the period $[t'_{i,j}, t_{i,j+1}-1]$ of additional $T(\epsilon_{K-i})$ steps, when
the discrepancy reduces with probability at least $1/2$. 
If the discrepancy remains unchanged, then there is another subinterval
and another chance, again with probability at least $1/2$, to reduce the discrepancy. 
The final subinterval $[t_{i,j_i-1},t_{i,j_i}-1]$ of interval $[t_i,t_{i+1}-1]$ 
ends at the last step when the discrepancy is still $K-i$; note that $t_{i,j_i} = t_{i+1}$.
The number $j_i\ge 1$ of subintervals of interval $[t_i,t_{i+1}-1]$ is a random variable.

We express the time $t_{K/4+1}$ when the discrepancy becomes $(\nicefrac{3}{4})K$ 
in the following way.
\begin{align}
t_{K/4+1} 
&= \sum_{i=0}^{K/4}\brac{t_{i+1} - t_i}
= \sum_{i=0}^{K/4}\brac{t_{i,j_i} - t_{i,0}}
= \sum_{i=0}^{K/4} \; \sum_{j=0}^{j_i-1} 
    \brac{\brac{t_{i,j+1} - t'_{i,j}} + \brac{t'_{i,j} - t_{i,j}}} \nonumber\\
&= \sum_{i=0}^{K/4} \; \sum_{j=0}^{j_i-1} 
    \brac{t_{i,j+1} - t'_{i,j}}
    \; +\; \sum_{i=0}^{K/4} \; \sum_{j=0}^{j_i-1} 
        \brac{t'_{i,j} - t_{i,j}}. \label{jdks6fa}
\end{align}
The expectation of the first double sum on line~\eqref{jdks6fa} is
\begin{align}
\E{\sum_{i=0}^{K/4} \sum_{j=0}^{j_i-1} 
    \brac{t_{i,j+1} - t'_{i,j}}}
&\; \le\; \sum_{i=0}^{K/4}
    T(\epsilon_{K-i})\cdot \E{{j_i}}
\;\le\; 2 \sum_{i=0}^{K/4}
    T(\epsilon_{K-i}).\label{vcb34a}
\end{align}
The first inequality above holds because 
for each $0\le i \le K/4$, $t_{i,j+1} - t'_{i,j}$ is equal to $T(\epsilon_{K-i})$
for each $0\le j< j_i-1$, and is at most $T(\epsilon_{K-i})$ for $j = j_i-1$
(that is, for the last subinterval).
The second inequality holds because for the number $j_i$ of subintervals 
of $[t_{i}, t_{i+1}-1]$, 
we have $\E{{j_i}}\le 2$:
each interval $[t_{i}, t_{i+1}-1]$ has the first subinterval and then each subsequent subinterval 
only with probability at most $1/2$.

Take $T=T_K=\left\lceil100Kn\log(Kn)/\Phi^2\right\rceil$ from \autoref{lem:smallK}. 
Since $K-k\le K$, this satisfies the lower bound on $T$ in \autoref{lemma:goodindexdrops2} 
with $c=1$ for every relevant index $k$. 
We define
$\tilde{T} = T + 4\cdot\sum_{i=0}^{K/4} T(\epsilon_{K-i})$ 
and show that $\pr{t_{K/4+1} \ge \tilde{T}} \le 2/3$.
From~\eqref{jdks6fa} and~\eqref{vcb34a}, we have
\begin{align}
\pr{t_{K/4+1} \ge \tilde{T}} 
&= \pr{[0,\tilde{T}-1] \subseteq [0,t_{K/4+1}-1]} 
= \pr{\,[0,\tilde T -1]\subseteq \bigcup_{i=0}^{K/4} \: \bigcup_{j=0}^{j_i-1} 
    \left[t_{i,j}, t_{i,j+1}-1\right] }
\nonumber\\
&\le\; \pr{\:\left| [0,\tilde T -1] \cap \bigcup_{i=0}^{K/4} \: \bigcup_{j=0}^{j_i-1} 
    \left[t_{i,j}, t'_{i,j}-1\right]\right| \ge T \;\; \text{and} \;\: t_{K/4+1} \ge \tilde{T}}
    \label{kl4ds2a1}\\
&\;\;\;\;\; + \pr{\: \left|[0,\tilde T -1] \cap \bigcup_{i=0}^{K/4} \: \bigcup_{j=0}^{j_i-1} 
    \left[t'_{i,j}, t_{i,j+1}-1\right]\right| \ge 4\cdot\sum_{i=0}^{K/4} T(\epsilon_{K-i}) } \nonumber\\
&\le\; \pr{\,\abs{\{t\in[\tilde{T}-1]: \tilde\eta(t)>\epsilon_K m\}} \ge T\,
\; \text{and} \; t_{K/4} \ge \tilde{T}} \label{kl4ds2a2}\\
&\;\;\;\;\; +\; \pr{\,\sum_{i=0}^{K/4} \: \sum_{j=0}^{j_i-1} 
    \brac{t_{i,j+1} - t'_{i,j}} \ge
4\cdot\!\!\sum_{i=0}^{K/4}{T}(\epsilon_{K-i})}.
 \label{kl4ds2a2X}
\end{align}
The probability given on line~\eqref{kl4ds2a2} 
upper bounds the probability on line~\eqref{kl4ds2a1} because for each 
$0 \le i \le K/4$ and $0\le j \le j_i-1$,
in each step $t\in [t_{i,j},t'_{i,j}-1]$, we have $\tilde\eta(t) > \epsilon_{K-i} m \ge \epsilon_K m$.
From~\eqref{vcb34a}, the probability given on line~\eqref{kl4ds2a2X}
is at most $1/2$:
the probability that the value of a positive random variable is at least twice its expectation 
is not greater than $1/2$.
We upper bound the probability~\eqref{kl4ds2a2} in the following way,
where $\Vec{X}(.)$, $\tau_k(.)$ and~$\mathcal{B}_k(.)$ refer to the original process,
and $\Vec{X'}(.)$, $\tau'_k(.)$ and~$\mathcal{B}'_k(.)$ refer to the mirrored process.
The terms $\tau_k(.)$ and $\mathcal{B}_k(.)$ are defined in \autoref{lemma:goodindexdrops2}.
\begin{align}
& \pr{\,\left|\{t\in[\tilde{T}-1]: \tilde\eta(t)>\epsilon_K m\}\right| \ge T\,
\; \text{and} \;
 t_{K/4+1} \ge \tilde{T}}\nonumber\\
&\;\;\; \le \textbf{Pr}\left[\,\left|\{t\in[\tilde{T}-1]: \tilde\eta(t)>\epsilon_K m\}\right| \ge T \right. \nonumber\\
&\;\;\;\;\;\;\;\;\;\;\;\;\;
\;\: \text{and} \;
\left. \left\{\, \forall k\in[K/2+1, (\nicefrac{3}{4})K]: \, 
\left(\Psi^{(+)}_k\left(\Vec{X}({\tilde{T}})\right) \ge 1 \, \; \text{and} \;\:
\Psi^{(+)}_k\left(\Vec{X'}({\tilde{T}})\right) \ge 1 \right)\right\}\right]\,\nonumber\\
&\;\;\; \le \textbf{Pr}\left[\, \left\{\, \exists k \in \left[K/2+1,(\nicefrac{3}{4})K\right]\,:\,
\left(\tau_k(\tilde T -1) \ge T/4 \, \; \text{or} \;\: \tau'_k(\tilde T -1) \ge T/4\right) \;
\right\}\right. \label{kl3a}\\
&\;\;\;\;\;\;\;\;\;\;\;\;\;\; \text{and} \,
\left. \left\{\, \forall k\in[K/2+1, (\nicefrac{3}{4})K]:\,
\left( \Psi^{(+)}_k\left(\Vec{X}({\tilde{T}})\right) \ge 1 \,\; \text{and} \;\:
\Psi^{(+)}_k\left(\Vec{X'}({\tilde{T}})\right) \ge 1\right)
\right\}\, \right]\nonumber\\
&\;\;\; \le \textbf{Pr}\left[\, \exists k \in \left[K/2+1,(\nicefrac{3}{4})K\right]:\,
\right. \nonumber\\
&\;\;\;\;\;\;\;\;\;\;\;\;\;\;\;\;\; \left. 
\left(\left(\tau_k(\tilde T -1) \ge T/4 \; \text{and} \;
\Psi^{(+)}_k\left(\Vec{X}({\tilde{T}})\right) \ge 1\right)\; \text{or} \;
\left(\tau'_k(\tilde T -1) \ge T/4 \; \text{and} \;
\Psi^{(+)}_k\left(\Vec{X'}({\tilde{T}})\right) \ge 1\right)\right)
\right]\nonumber\\
&\;\;\; = \textbf{Pr}\left[\, \exists k \in \left[K/2+1,(\nicefrac{3}{4})K\right]:\,
\left( \mathcal{B}_k(\tilde T) \; \text{or} \; \mathcal{B}'_k(\tilde T) \right)
\right]\nonumber\\
&\;\;\; \le\;
\sum_{k=K/2+1}^{(\nicefrac{3}{4})K} \left(\pr{\mathcal{B}_k(\tilde T)} + \pr{\mathcal{B}'_k(\tilde T)}\right)
\; \le\; \frac{K+1}{4}\cdot\frac{2}{(K/4)^{3} n^{3}}
\;=\; O\left(\frac{1}{n^{3}}\right).\label{dfgg3s}
\end{align}
Inequality~\eqref{kl3a} follows from \autoref{lemma:goodindexexists2}.
The inequalities on line~\eqref{dfgg3s} follow from the union bound and \autoref{lemma:goodindexdrops2}.
The bound of $1/2$ on the probability~\eqref{kl4ds2a2X} and the bound~\eqref{dfgg3s}
imply that $\pr{t_{K/4+1} \ge \tilde{T}}$ is at most $2/3$, 
so with probability at least $1/3$, in $\tilde{T}$ steps
the discrepancy is reduced to at most $(\nicefrac{3}{4})K$.
Using the argument
of restarting in case of failure, 
we conclude that the expected time to reduce the discrepancy to $(\nicefrac{3}{4})K$
is $O(\tilde{T})$.
\end{proof}

\begin{proof}[Proof of \autoref{thm:main:version2}]
Using \autoref{lem:smallK},
the expected completion time in the \ProtAV process
is at most of the following order.
\begin{align}
\sum_{i=0}^{\log_{4/3}K} \left(T_{(\nicefrac{3}{4})^i K} \,
+ \sum_{k=(\nicefrac{3}{4})^{i+1}K}^{(\nicefrac{3}{4})^{i}K} {T}(\epsilon_k)\right)
&\le \sum_{i=0}^{\log_{4/3}K} 100\frac{(\nicefrac{3}{4})^i K n \log(Kn)}{\Phi^2}
+ \sum_{k=0}^{K} {T}(\epsilon_k) = O\left( \frac{K n \log(Kn)}{\Phi^2}\right)
+ \sum_{k=0}^{K} {T}(\epsilon_k). \label{yyyfdd}
\end{align}
For the final sum above, the bound 
$T(\epsilon)=O(\epsilon\gamma n^2/\Phi)$
(from \autoref{lem:dropbyone}) gives
\begin{equation}\label{xxz3z}
\sum_{k=0}^{K} {T}(\epsilon_k)
\; = \; O\left(\frac{\gamma n^2}{\Phi}\sum_{k=0}^{K} {\epsilon_k}\right),
\end{equation}
and we have
\begin{equation}\label{gg5ff}
\sum_{k=0}^{K}{\epsilon_k}
= \sum_{k=0}^{K} e^{-\Phi k/16}
\le \frac{1}{1-e^{-\Phi/16}}
= O\left(\frac{1}{\Phi}\right).
\end{equation}
Putting together~\eqref{yyyfdd}, \eqref{xxz3z} and~\eqref{gg5ff},
we get the bound on the expected completion time in the \ProtAV process stated in \autoref{thm:main:version2}.
\end{proof}

\section{Proof of Theorem~\ref{thm:lowerbound}}
\label{sec:lowerboundproof}
We show a lower bound on the convergence time of DIV in 
some regular graphs $G=(V,E)$ as stated in 
\Autoref{thm:lowerbound}.
For a regular graph, one step of DIV, 
when a random node $v$ and then its random neighbor $u$ are selected and the opinion of $v$ is updated,
is equivalent to selecting a random edge $\{v,u\}\in E$ and
deciding by a fair coin flip whether~$v$ or~$u$ is updated.

We first prove two lemmas.
\Autoref{lem:delta_mean_bound} gives us some lower bound
on the expected one-step progress of DIV, which is then used
in \Autoref{lem:lowbound72} to show a lower bound on the 
number of steps needed to ensure 
'good' probability of 'substantial' progress towards
convergence.
We will conclude the proof of \Autoref{thm:lowerbound}
by constructing a family of regular graphs such that 
when applying \Autoref{lem:lowbound72}, we will get 
the lower bound claimed in this theorem.

Given a nonempty set of nodes $S \subseteq V$ and a vector of opinions $\vx = (x_i)_{i \in V}$,
    we write $\mu_S(\vx) \coloneqq \frac{1}{\abs{S}} \sum_{i \in S} x_i$ for the mean opinion in $S$.
Furthermore,
    we write $\Delta_S(\vx) \coloneqq \mu_S(\vx) - \mu_{\overline{S}}(\vx)$ for the difference between the mean opinions in $S$ and in its complement $\ovl S$.

\begin{lemma}\label{lem:delta_mean_bound}
    Let $S \subseteq V$ be a nonempty subset of $V$ of size $\abs{S} \leq n / 2$.
    Then for all $\vx \in \mathds{N}^V$ and $t \geq 1$, we have
\begin{equation}\label{kl2s1}
\abs*{\Delta_S(\vX{t}) - \Delta_S(\vX{t-1})} 
\leq \frac{1}{\abs{S}},
\end{equation}
\begin{equation}\label{kl2s2}
\abs*{\E{\Delta_S(\vX{t}) - \Delta_S(\vX{t-1}) \mid \vX{t-1} = \vx}}
\leq \frac{\abs*{\mathbb{O}(S)}}{\abs{E}\, \abs{S}}.
\end{equation}
\end{lemma}

\begin{proof}
Inequality~\eqref{kl2s1} holds because at most one node can change its opinion in one step,
implying that only one of $\mu_S$ and $\mu_{\overline{S}}$ can change at this step.
If $\mu_S$ changes, then this change is $\pm \frac{1}{\abs{S}}$; if $\mu_{\overline{S}}$ changes,
then this change is $\pm \frac{1}{n-\abs{S}}$, so the absolute change is never greater than $\frac{1}{\abs{S}}$
(as $\abs{S} \leq n/2$).

To prove~\eqref{kl2s2}, condition on $\vX{t-1}=\vx$ and consider the signed change in $\mu_S$.
For an interaction within $S$, the two possible updates have opposite effects on its mean
(or both leave it unchanged), so their expected contribution is zero.
An interaction entirely outside $S$ leaves its mean unchanged.
Thus only edges crossing the cut contribute to the conditional drift, giving
\[
\E{\mu_S(\vX{t})-\mu_S(\vX{t-1}) \mid \vX{t-1}=\vx}
=\frac{1}{2\abs{E}\abs{S}}
\sum_{\substack{\{i,j\}\in\mathbb{O}(S)\\ i\in S}}\sign(x_j-x_i),
\]
where $\sign(z)$ is $1$, $-1$, or $0$ according as $z$ is positive, negative, or zero.
Applying the same identity to $\overline S$ and using $\sign(x_i-x_j)=-\sign(x_j-x_i)$,
linearity of expectation yields
\begin{align*}
&\E{\Delta_S(\vX{t})-\Delta_S(\vX{t-1}) \mid \vX{t-1}=\vx}\\
&\qquad=\frac{1}{2\abs{E}}\left(\frac{1}{\abs{S}}+\frac{1}{n-\abs{S}}\right)
\sum_{\substack{\{i,j\}\in\mathbb{O}(S)\\ i\in S}}\sign(x_j-x_i).
\end{align*}
Taking the absolute value of this conditional expectation, and using
$\abs{\sign(z)}\le1$ and $\abs{S}\le n/2$, gives
\begin{align*}
\abs*{\E{\Delta_S(\vX{t})-\Delta_S(\vX{t-1}) \mid \vX{t-1}=\vx}}
&\le\frac{\abs{\mathbb{O}(S)}}{2\abs{E}}
\left(\frac{1}{\abs{S}}+\frac{1}{n-\abs{S}}\right)\\
&\le\frac{\abs{\mathbb{O}(S)}}{\abs{E}\abs{S}},
\end{align*}
as required.
\end{proof}

\begin{lemma}\label{lem:lowbound72}
    Let $S \subseteq V$ be a nonempty subset of $\,V$ having size $\abs{S} \leq n / 2$,
        and let $\vec{X}(0)$ be an initial configuration
        such that $\Delta_S(\vX{0}) = K$.
    Then for
    $T = \min\left\{{\frac{K \abs{E} \abs{S}}{4 \abs*{\mathbb{O}(S)}}, \frac{K^2 \abs{S}^2}
    {128\ln 2}}\right\}$, we have
\[\pr{\Delta_S(\vX{T}) \le {K}/{2}} \leq \frac{1}{2}.\]
\end{lemma}

\begin{proof}
For $t\ge 0$, let
\[ Z_t \coloneqq \Delta_S(\vX{t}) + t \cdot \frac{\abs{\mathbb{O}(S)}}{\abs{E} \, \abs{S}}.
\]

Let $\mathcal F_t$ record the interactions before step $t$.
By~\eqref{kl2s2} and the Markov property, for every $t\ge1$,
\[
\E{\Delta_S(\vX{t})-\Delta_S(\vX{t-1})\mid\mathcal F_{t-1}}
\ge-\frac{\abs{\mathbb{O}(S)}}{\abs{E}\abs{S}}.
\]
Consequently,
\begin{align*}
\E{Z_t-Z_{t-1}\mid\mathcal F_{t-1}}
&=\E{\Delta_S(\vX{t})-\Delta_S(\vX{t-1})\mid\mathcal F_{t-1}}
+\frac{\abs{\mathbb{O}(S)}}{\abs{E}\abs{S}} \;\ge0.
\end{align*}
Since $(Z_t)_{t\ge0}$ is adapted to $(\mathcal F_t)_{t\ge0}$ and each $Z_t$ is integrable,
it is a submartingale.

From~\eqref{kl2s1}, we have the following bound on the one-step differences in $(Z_t)_{t\ge 0}$.
\[ \abs*{Z_t - Z_{t-1}} \;=\; \abs*{\Delta_S(\vX{t}) - \Delta_S(\vX{t-1})
+ \frac{\abs*{\mathbb{O}(S)}}{\abs{E}\, \abs{S}}}
\;\le\; \frac{1}{\abs{S}} + \frac{\abs*{\mathbb{O}(S)}}{\abs{E}\, \abs{S}}
\;\le\; \frac{2}{\abs{S}}.
\]
We apply the one-sided Azuma--Hoeffding inequality for submartingales (see \autoref{AzumaHoeffding}),
taking $\delta = \sqrt{8T \ln 2} / \abs{S}$, so that 
\[ 
T \cdot \frac{\abs{\mathbb{O}(S)}}{\abs{E}\abs{S}} + \delta
\;\le\; \frac{K \abs{E} \abs{S}}{4 \abs*{\mathbb{O}(S)}} \cdot \frac{\abs{\mathbb{O}(S)}}{\abs{E}\abs{S}}
    \;+\; \left(\sqrt{8\frac{K^2 \abs{S}^2}{128\ln 2}\ln 2}\right)/|S|
\; = \; \frac{K}{2}.
\]
We conclude the proof of the lemma, using Azuma--Hoeffding Inequality~\eqref{AZHofsubm}
in line~\eqref{jlfh7s} below.
    \begin{align}
    \pr{\Delta_S(\vX{T})\le K/2}
    &= \pr{\Delta_S(\vX{T})\le \Delta_S(\vX{0}) - K/2}\nonumber\\
    &\le \pr{\Delta_S(\vX{T}) \le \Delta_S(\vX{0}) - \left({T} \cdot \frac{\abs{\mathbb{O}(S)}}{\abs{E}\abs{S}} 
    + \delta\right)}\nonumber\\
    &= \pr{Z_T \le Z_0 - \delta} 
    \; \leq\;  \exp\left(-\,\frac{\delta^2}{2T (2 / \abs{S})^2}\right)
        = \exp\left(-\,\frac{\delta^2 \cdot \abs{S}^2}{8 T}\right) = \frac{1}{2}.\label{jlfh7s}
    \end{align}
\end{proof}

We conclude the proof of \autoref{thm:lowerbound} with the following graph construction.
For integers $r\ge2$ and $1\le q\le\lfloor r/2\rfloor$, take two disjoint cliques of size $r$ and join them by $q$ edge-disjoint perfect matchings.
Let $G=(V,E)$ be the resulting graph, let $n=2r$, and let $U\subseteq V$ be the node set of one clique.
The graph is regular of degree $d=r-1+q$, 
and has $m=rd$ edges.
The cut $(U,V\setminus U)$ consists exactly of the $qr$ matching edges between the cliques, so
\[
|U|=|V\setminus U|=r,\qquad
|\mathbb{O}(U)|=qr,\qquad
\vol{U}=rd=m.
\]
This cut attains the conductance of $G$.
Indeed, consider any nonempty set $S$ with $\vol{S}\le m$, equivalently $|S|\le r$.
Let $a$ and $b$ be its numbers of nodes in the two cliques, ordered so that $a\ge b$.
The clique edges and the matching edges give
\[
|E(S,V\setminus S)|
\ge a(r-a)+b(r-b)+q(a-b)
\ge q(a+b),
\]
where $a(r-a)+b(r-b)\ge br\ge2qb$ follows from $a+b\le r$, $a\ge b$, and $2q\le r$.
Dividing by $\vol{S}=d(a+b)$ shows that every such set has conductance at least $q/d$, with equality for $S=U$
(when $a=r$ and $b=0$).
Consequently,
\[
\Phi(G)=\frac{|\mathbb{O}(U)|}{\vol{U}}
=\frac{q}{r-1+q}=\Theta(q/n)\ge\frac{1}{n}.
\]
Taking $q=1$ gives conductance $\Theta(1/n)$, 
$q=\Theta(n)$ gives constant conductance, and varying $q$ realizes all intermediate conductance scales.

We now apply \autoref{lem:lowbound72} with $S=U$.
In the initial configuration $\vec{X}(0)=\Vec{x}$, 
all nodes in $U$ have opinion $K$ and all nodes in $V\setminus U$ have opinion $0$,
so both $\text{disc}(\Vec{x})$ and $\Delta_S(\Vec{x})$ 
are equal to $K$.
With this setting, the time threshold defined in \autoref{lem:lowbound72} is
\[
T=\min\left\{\frac{Kn}{8\Phi(G)},\frac{K^2n^2}{512\ln 2}\right\}
=\Omega(Kn/\Phi(G)),
\]
since $K\ge1$ and $\Phi(G)\ge1/n$ imply $K^2n^2\ge Kn/\Phi(G)$.
Thus we have
\[
\pr{T_G(\Vec{x}) \le T}
\le \pr{\disc{\vX{T}} \le {K}/{2}}
\le \pr{\Delta_U(\vX{T}) \le {K}/{2}} \leq \frac{1}{2},
\]
implying that 
$\E{T_G(\Vec{x})} \in \Omega\!\left({K n}/{\Phi(G)}\right).$

There is also a binary initial configuration $\Vec{y}$ with discrepancy at most $1\le K$ such that $\E{T_G(\Vec{y})}=\mathcal{T}^{2V}_G$. Maximizing over all initial configurations of discrepancy at most $K$ therefore gives at least the maximum of these two lower bounds, and $\max{}\{a,b\}\ge(a+b)/2$ gives the combined bound in~\autoref{thm:lowerbound}.

\section{Proof of Theorem~\ref{thm:main2}: Reducing the Discrepancy from \texorpdfstring{$O(\log n)$ to $2$}{O(log n) to 2}
in Expanders}
\label{sec:disc3}

\Autoref{thm:main2} considers expanders 
paramiterized by $\lambda(G)$, and
the proof of this theorem is
described in \Autoref{sec:outlines-of-proofs}.
The reduction of 
the discrepancy to $O(\log n)$ refers to
what is proven for this part of the process in 
\Autoref{thm:main} for the general graphs parameterized by
their conductance.
The reduction of the discrepancy from $O(\log n)$ to $2$
refers to \Autoref{lem:one-step}, 
which bounds the number of steps needed
to claim that with probability at least $1/3$,
the discrepancy has 
reduced (at least by $1$),
and to \Autoref{lem:k-step},
which uses \Autoref{lem:one-step} iteratively
to give a bound on the time to reduce 
the discrepancy to $2$. 

\Autoref{lem:one-step} and \Autoref{lem:k-step}
are the only missing ingredients in the proof of  
\Autoref{thm:main2}.
In this section, we provide their formal statements 
and proofs 
(proving first \Autoref{lem:k-step}
and then \Autoref{lem:one-step}).

\begin{proposition}
\label{lem:one-step}
Consider the \ProtAV process on a graph $G$ with $\gamma(G)=O(1)$ and 
the second largest absolute eigenvalue 
$\lambda(G) \le \xi<1$, for a constant $\xi$.
For any initial configuration with discrepancy $K \geq 3$, with probability at least $\frac{1}{3}$,  
the discrepancy is reduced by one within $O(\lambda(G)^2\cdot n^2 + n\log n)$ steps.
This implies (by the argument of 'restarting on failure')
the same asymptotic bound on the expected time to reduce the discrepancy by one.
\end{proposition}

\begin{corollary}
\label{lem:k-step}
Consider the \ProtAV process on a graph $G$ 
with $\gamma(G)=O(1)$ and the second largest absolute eigenvalue $\lambda(G) \le \xi$ for 
a constant $\xi<1$.
Assume $\vec{x}(0)$ is a configuration with discrepancy $K\ge 3$.
Then the time $T^{(2)} = T_{G}^{(2)}(\vec{x}(0))$ 
to reduce the discrepancy to $2$ is
in expectation $O\!\left( K \cdot \left(\lambda(G)^2 \cdot n^2 + n\log n\right) \right)$
and \whp $O\!\left( (K+\log n) \cdot \left(\lambda(G)^2 \cdot n^2 + n\log n\right) \right)$.
\end{corollary}

\begin{proof}[Proof of \autoref{lem:k-step}]
The 'in expectation' bound follows immediately from 
\autoref{lem:one-step} by the linearity of expectation.
The '\whp' bound follows from \autoref{lem:one-step} and the concentration of the binomial distribution,
as explained below.

By \autoref{lem:one-step}, there is a $\beta \in O(\lambda(G)^2\,n^2 + n\log n)$ such that
for each time step $t\ge 0$,  
if the discrepancy at time $t$ is at least $3$, then regardless of the past 
the probability that within the next $\beta$ steps the discrepancy decreases (by $1$ or more) 
is at least $1/3$.
Partition time into phases of length $\beta$ and let $Z_i = 1$, 
if the discrepancy decreases during phase $i$ 
or it is already less than $3$ at the beginning of this phase;
otherwise $Z_i = 0$.
Thus, whatever the values of $Z_1,Z_2, \ldots Z_{i-1}$ are,
we always have ${\pr{Z_i = 1 \mid Z_1,\ldots,Z_{i-1}} \ge 1/3}$.
This means that 
we can couple the sequence $(Z_i)_{i\ge 1}$ with a sequence $(Y_i)_{i\ge 1}$ of 
i.i.d. $0$-$1$ Bernoulli variables with the probability of success $1/3$
such that for each $i\ge 1$, $Z_i \ge Y_i$.
Denoting $L^{(Z)} = \min{}\{\ell\ge 0: \sum_{i=1}^\ell Z_i \ge K\}$
and $L^{(Y)} = \min{}\{\ell\ge 0: \sum_{i=1}^\ell Y_i \ge K\}$, 
we have ${T^{(2)}} \le \beta\cdot L^{(Z)} \le \beta\cdot L^{(Y)}$,
implying $\E{T^{(2)}} \le \beta\cdot\E{L^{(Y)}} = 3K\beta$.

To show the \whp bound on ${T^{(2)}}$,
take $L=3h\log n+6(h+1)K$ for an arbitrary positive constant $h$, and let $S_L=Y_1+\cdots+Y_L$. 
We have
\begin{align*}
\pr{T^{(2)} > L\beta} 
&= \pr{\disc{\Vec{X}(L\beta)} \ge 3} 
\le \pr{\sum_{i=1}^L Z_i \le K} \\
&\le \pr{S_L \le K} = \pr{S_L \le L/3 -(L/3-K)} \le \exp\brac{-\frac{(L/3 - K)^2}{2\cdot L/3}}
 \le \exp\brac{-\frac{L/3 - 2K}{2}}
 \le n^{-h/2}.\label{ghjk21la}
\end{align*}
We applied above 
the Chernoff bound 
given in \autoref{Chernoff} to sequence ${(Y_i)_{i=1}^L}$
with $\mu = \E{S_L} = L/3$ 
and $\delta = L/3 -K$. 
Thus \whp the discrepancy drops below $3$ in 
$L\beta = O\!\left( (K+\log n) \cdot \left(\lambda(G)^2 n^2 + n\log n\right) \right)$ steps.
\end{proof}

\begin{proof}[Proof of \autoref{lem:one-step}]
To prove the bound stated in \autoref{lem:one-step}, we use a combination of \autoref{lemma:goodindexdrops} and \autoref{lem:dropbyone}.

Apply \autoref{lemma:goodindexdrops} with $k=K-1$, $c=1$ and $T=\left\lceil100\tfrac{n\log n}{\Phi^2(G)}\right\rceil$, which is $O(n\log n)$
since $\lambda(G) \le \xi<1$ implies $\Phi(G)=\Theta(1)$.
By \autoref{lemma:goodindexdrops}, this $T$ guarantees that \whp
opinion $K$ disappears by step $T$ \textbf{unless} index $K-1$ is good 
(in the sense of \autoref{def:kgood}) in fewer than
a quarter of the steps in $[0,T-1]$. 
By symmetry of the mirrored process, 
an analogous statement applies to opinion $0$:
\whp opinion $0$ disappears by step $T$ \textbf{unless} index $K-1$ is good in fewer than
a quarter of the steps in $[0,T-1]$ in the mirrored process.

We distinguish two cases based on the size of $\tilde\eta(t)$
during the interval $[0,T-1]$.
Recall that for a step $t \ge 0$,
$\tilde\eta(t) = \min\{\vol{S_0(t)}, \voll(S_K(t))\}$,
and $S_0(t)$ and $S_K(t)$ are the sets of nodes with opinions~$0$ and~$K$,
respectively, at this step,
as introduced where \autoref{lem:dropbyone} was discussed.

Let $\tau\ge 0$ be the first step when 
$\tilde\eta(t)\le 16(2m)\lambda(G)^2$,
and let $D\ge 0$ be the first step 
when the discrepancy drops:
\begin{align*}
\tau&=\inf\{t\ge0:\tilde\eta(t)\le16(2m)\lambda(G)^2\},\\
D&=\inf\{t\ge0:\disc{\X{t}}<K\}.
\end{align*}

\noindent
\underline{\textbf{Case 1:} \textit{ Suppose $\tilde\eta(t) >  16\cdot {{2m} \cdot \lambda(G)^2}$ for each $0\le t\le T-1$, that is, $\tau \ge T$.}}

{Suppose no discrepancy drop has occurred by time $T$, 
so the extreme opinions at step $T$ are still 
the initial extreme opinions $0$ and $K$. 
Consider an arbitrary step $t$, $0\le t\le T-1$,}
and show that
index $K-1$ is good in the original or mirrored process (or both).
As the discrepancy $K$ is at least $3$, 
then $2 \le K-1$, so 
\[ S_{K-1}(t) = \{i\in V: x_i(t) \ge K-1\} \subseteq \{i\in V: x_i(t) \ge 2\} = {S_{2}(t)},
\] 
implying that $\vol{S_{K-1}(t)}\le m$ or 
$\vol{\ovl S_{2}(t)}\le m$.
Note that the set $\ovl S_{2}(t)$ in the original process is the set $S'_{K-1}(t)$ in the mirrored process,
so what we mean here is that $\vol{S_{K-1}(t)}\le m$ in the original process or 
$\vol{S'_{K-1}(t)}\le m$ in the mirrored process.
Suppose w.l.o.g.\
the former, that is, $\vol{S_{K-1}(t)}\le m$, and show that index $K-1$ is good.
(Otherwise, by considering the mirrored process, 
show that index $K-1$ is good in the mirrored process.)
Following the notation~\eqref{defn-Splus}
from \autoref{def:kgood} of good indices (applied to index $K-1$ and step $t$), we set
\begin{align*}
    S(t) \coloneqq S_{K-1}(t),
    \;\;
    \overline{S}(t) \coloneqq V \setminus S(t) \:=\: \{ i \in V \mid x_i(t) \le K-2 \}.
\end{align*}
In order to show that $K-1$ is good, we must show
$$\abs{E(S_K(t),\ovl S(t))} \geq \tfrac{1}{4} \cdot \Phi(G)\cdot\vol{S_K(t)}.$$
The expander mixing lemma (cf. \autoref{Lex:expandermixing}) gives
\begin{align}\label{tyhs30as}
    \abs{E(S_K(t),\ovl S(t))}
\ge \; &
\frac{\vol{S_K(t)}\cdot\vol{\ovl S(t)}}{2m}
-
\lambda(G) \cdot\sqrt{\vol{S_K(t)}\cdot\vol{\ovl S(t)}}.
\end{align}
Using $\vol{S(t)} \le m$ and the assumption of this case, we obtain
\[
m \le \vol{\ovl S(t)} \leq 2m \le \tilde\eta(t)/(16\cdot \lambda(G)^2) 
\leq \vol{S_K(t)}/(16\cdot \lambda(G)^2).
\]
Inequality~\eqref{tyhs30as}, the above upper and lower bounds on $\vol{\ovl S(t)}$, 
and $\Phi(G) \le 1$, give
\begin{align*}
    \abs{E(S_K(t),\ovl S(t))}
&\ge
\frac{\vol{S_K(t)}}{2}
-
\lambda(G) \sqrt{ \frac{\vol{S_K(t)}^2}{16\cdot \lambda(G)^{2} }} = \frac{\vol{S_K(t)}}{4}
\ge \frac{1}{4}\cdot \Phi(G)\cdot\vol{S_K(t)}
\end{align*}
Thus index $K-1$ is good in step $t$.
We conclude that 
in each step $t$ index $K-1$ is good in the original process or index $K-1$ is good in the mirrored 
process, so index $K-1$ is good in at least half of the steps in $[T-1]$ in the original process, or
index $K-1$ is good in at least half of the steps in $[T-1]$ in the mirrored process.

{Thus, 
referring to \autoref{lemma:goodindexdrops}, the event that Case 1 holds ($\tau \ge T$) and no discrepancy drop 
has occurred by time $T$ ($D > T$)
is contained in $\mathcal{B}_{K-1}(T)\cup\mathcal{B}'_{K-1}(T)$, 
so the probability of this event is $O(n^{-3})$:}
\begin{align*}
    \pr{D>T \; \text{and} \; \tau\ge T} = O(n^{-3}).
\end{align*}

\noindent
\underline{\textbf{Case 2:} \textit{Suppose $\tilde\eta(t) \le 16\cdot {{2m} \cdot \lambda(G)^2}$
for some $t$, $0\le t\le T-1$, that is, $0 \le \tau \le T-1$.}}

Fix an arbitrary $t\in [T-1]$, and assume that $\tau = t$
and the discrepancy has not dropped by step $t$. Apply \autoref{lem:dropbyone} 
at this step to conclude that 
the discrepancy drops with probability at least $1/2$ within additional $T'$ steps, where
\[
T'=O\left(\frac{ \tilde\eta(\tau) \cdot n}{\Phi(G) \cdot \dmin}\right)
= O\left(\frac{\lambda(G)^2 mn}{\Phi(G)\dmin}\right)  = O\left(\lambda(G)^2n^2\right).
\]
Here Cheeger's inequality (\autoref{Cheeger}) 
gives $\Phi(G)\ge(1-\lambda(G))/2=\Omega(1)$, and $\gamma(G)=O(1)$ gives $\dmin=\Omega(m/n)$.
We choose $T'$ uniformly over all 
configurations satisfying the bound on the support for the
extreme opinions, so
\begin{align*}
    \pr{D>t+T'|\, \tau=t} \le \frac{1}{2}.
\end{align*}

Combining both cases,
the following bound on the 
probability that the discrepancy does not drop 
within $T+T'$ steps holds for sufficiently large $n$.
\begin{align*}
\pr{D>T+T'}
&= \pr{D>T+T'\; \text{and} \; \tau\ge T}
    + \sum_{t=0}^{T-1} \pr{D>T+T'|\, \tau=t} \cdot \pr{\tau=t}\\
&\le \pr{D>T \; \text{and} \; \tau\ge T}
    + \sum_{t=0}^{T-1} \pr{D>t+T'|\, \tau=t} \cdot \pr{\tau=t}
    \; \le \; O(n^{-3}) + \frac{1}{2} \le \frac{2}{3}.
\end{align*}
Hence the discrepancy drops with probability at least $1/3$ within
$T+T'=O(n\log n+\lambda(G)^2n^2)$ steps, as required.
\end{proof}

\section{Omitted Proofs}

\subsection{Proof of \autoref{Claim1}}
\label{sec:proofClaim1}

We consider one step $t\ge 0$ and probabilities and expectations conditional on 
$\vX{t}=\vec{x}=(x_1,x_2,\ldots,x_n)$.
For two nodes $i$ and $j$ define 
\begin{align*}
    \iota(i,j) \coloneqq \begin{cases}
        0, & \textit{if } x_i=x_j,\\
        -1, & \textit{if } x_i>x_j,\\
        1, & \textit{if } x_i<x_j.\\
    \end{cases}
\end{align*}
According to Protocol~\ref{prot:edgeinteraction}, 
at a given step a random directed edge $e=(u,v)$ is chosen as follows. 
Firstly a node $u$ is chosen uniformly at random, 
and then a neighbor $v \in N_u$ is chosen 
with probability $1/d_u$.  
Thus for a given node $i \in V$, 
\begin{align*}
& \E{d_i \cdot f(X_i(t+1))^2} \\
&\;\;\; = \pr{u \ne i} \cdot d_i \cdot f(x_i)^2 \; +\; \pr{u=i} \cdot d_i \cdot\!\! \sum_{j:\{i,j\}\in E}\pr{v=j \mid u=i} \cdot \E{ f(X_i(t+1))^2 \mid (u,v)=(i,j)}.
\end{align*} 
Since $u$ is uniform on $V$ and, conditional on $u$, $v$ is uniform on $N(u)$,
for a given node $i$, $\pr{u=i} = \frac{1}{n}$ and  
$\pr{v=j \mid u=i} = \tfrac{1}{d_i}$.
Furthermore, if $(u,v)=(i,j)$, then $f(X_i(t+1))=f(x_i+\iota(i,j))$, 
which is $0$ for $i \in \overline{S}$
(equivalently, for $x_i \le k-1$).
Recall that $f(x)= \text{max}\{x-k, 0\}$.
Thus, 
\begin{align}
    \E{\sum_{i \in V} d_i \cdot f(X_i(t+1))^2 } &= \left(1-\frac{1}{n}\right)\sum_{i \in V}  d_i \cdot f(x_i)^2 
   \; + \; \frac{1}{n} \sum_{i \in V} d_i \cdot {\sum_{j:\{i,j\} \in E} \tfrac{1}{d_i}  \cdot f(x_i + \iota(i,j))^2} \nonumber \\ 
   &= \left(1-\frac{1}{n}\right)\sum_{i \in V}  d_i \cdot f(x_i)^2 
   + \frac{1}{n} \underbrace{\sum_{(i,j) \in \Vec{E}} \ f(x_i + \iota(i,j))^2}_{(*)},
    \label{eqn:total_exp}
\end{align}

In the following, we focus on the second sum $(*)$, which is over {\emph{directed}} edges $(i,j)$. 
We rearrange this sum
by combining the two terms for the two opposite \emph{directed} edges $(i,j)$ and $(j,i)$
into one summand for the {\emph{undirected}} edge $\{i,j\}$,
noting that both $f(x_i + \iota(i,j))$ and $f(x_j + \iota(j,i))$ are equal to $0$ unless
at least one of $x_i$ and $x_j$ is greater than $k$
and that $\iota(i,j) = - \iota(j,i)$,
and recalling that $S_{+1} = \{ i \in V \mid x_i \geq k+1\}$,
$E_{S_{+1}} = \{\{i,j\}\in E: \text{$i$ or $j$, or both, in $S_{+1}$}\}$.
\begin{align*}
    (*) = \sum_{(i,j) \in \Vec E}  f(x_i + \iota(i,j))^2 
    &= \sum_{\substack{\{i,j\} \in E \\ \max{}\{x_i,x_j\} \geq k+1}}  \left(f(x_i + \iota(i,j))^2 
          + f(x_j + \iota(j,i))^2 \right) \\
    &=  \sum_{\{i,j\} \in E_{S_{+1}}} \left(f(x_i + \iota(i,j))^2 + f(x_j - \iota(i,j))^2\right).
\end{align*}
We consider a single summand in the right-hand sum above, 
for an (undirected) edge $\{i,j\}\in E_{S_{+1}}$, and show that 
\begin{align}
f(x_i + \iota(i,j))^2 + f(x_j - \iota(i,j))^2 \leq f(x_i)^2 + f(x_j)^2 - \abs*{f(x_i) - f(x_j)} + \mathds{1}_{|x_i-x_j|=1}.   \label{lem:edgediff}
\end{align}
Assuming w.l.o.g.\ that $x_i \ge x_j$, which implies that $x_i \ge k+1$,
inequality~\eqref{lem:edgediff} becomes
\begin{align}
f(x_i + \iota(i,j))^2 + f(x_j - \iota(i,j))^2 \leq f(x_i)^2 + f(x_j)^2 - (f(x_i) - f(x_j)) + \mathds{1}_{|x_i-x_j|=1}.    \label{lem:edgediffX}
\end{align}
We show~\eqref{lem:edgediffX} by considering the following  three cases.

\noindent \textbf{Case 1}: $x_i > x_j+1$.
    In this case, we have $\mathds{1}_{|x_i-x_j|=1}=0$ and $\iota(i,j)=-1$, {\it i.e.} the LHS in~\eqref{lem:edgediffX}
    is equal to $f(x_i-1)^2 + f(x_j + 1)^2$.
    Since $x_i \ge k+1$, we have
    \begin{align}
        f(x_i-1)^2 &= \big((x_i-1) - k \big)^2 = \big((x_i - k) - 1\big)^2
        = (f(x_i)-1)^2 = f(x_i)^2 - 2f(x_i) + 1. \label{kll39s}
    \end{align}
    Now, we distinguish between two subcases:
    \begin{itemize}
        \item If $x_j < k$, it holds $f(x_j ) = f(x_j + 1) = 0$, so
        \begin{align*}
            f(x_j + 1)^2 = 0 = f(x_j)^2 + f(x_j).
        \end{align*}
        This and \eqref{kll39s}  give~\eqref{lem:edgediffX}:
        \begin{align*}
            f(x_i -1)^2 + f(x_j +1)^2 
            &= f(x_i)^2 - 2f(x_i) + 1 + f(x_j)^2 + f(x_j)
            \leq f(x_i)^2 + f(x_j)^2 - (f(x_i) - f(x_j)),
        \end{align*}
        where the last inequality follows from $f(x_i) \geq 1$ (as $x_i \ge k+1$).
        \item Otherwise, if $x_j \geq k$, it holds that:
    \begin{align*}
        f(x_j+1)^2 &= \big((x_j+1) - k \big)^2 = \big((x_j - k) + 1\big)^2
        = (f(x_j)+1)^2 = f(x_j)^2 + 2f(x_j) + 1
    \end{align*}
    Thus, \begin{align*}
        f(x_i-1)^2 + f(x_j+1)^2&= f(x_i)^2 + f(x_j)^2   - 2f(x_i) + 2f(x_j) + 2
    \end{align*}
    As  $f(x_i)-f(x_j)=x_i-x_j \ge 2$, this simplifies to
        \begin{align*}
        f(x_i-1)^2 + f(x_j+1)^2&{\; \le} f(x_i)^2 + f(x_j)^2   - (f(x_i) - f(x_j)),
    \end{align*}
    matching~\eqref{lem:edgediffX}.
    \end{itemize}

\noindent \textbf{Case 2}: $x_i = x_j+1$; the case with $\mathds{1}_{|x_i-x_j|=1} = 1$. 
In this case, we also have $\iota(i,j)=-1$, so the LHS in~\eqref{lem:edgediffX}
is again equal to $f(x_i-1)^2 + f(x_j + 1)^2$, but now
\begin{align*}
    f(x_i-1) &= f(x_j)\; \text{ and } \;
    f(x_j+1) = f(x_{i}),
\end{align*}
giving
\begin{align*}
        f(x_i-1)^2 + f(x_j+1)^2&= f(x_i)^2 + f(x_j)^2.
    \end{align*}
By noting that $f(x_i)-f(x_j)= (x_i-k) - (x_j-k) = 1$, we get~\eqref{lem:edgediff}, with equality: 
\begin{align*}
f(x_i-1)^2 + f(x_j+1)^2 =  f(x_i)^2 + f(x_j)^2 - (f(x_i) - f(x_j)) + 1.
\end{align*}

\noindent \textbf{Case 3}: $x_i = x_j$. 
In this case, the two sides of~\eqref{lem:edgediffX} are equal because
$\iota(i,j)=\iota(j,i)=0$, $f(x_i)=f(x_j)$, and $\mathds{1}_{|x_i-x_j|=1} = 0$.

Thus Inequality~\eqref{lem:edgediffX} holds in all cases, that is,  
whenever $x_i\ge x_j$ and $x_i\ge k+1$,
implying that Inequality~\eqref{lem:edgediff} holds 
for each edge $\{i,j\}\in E_{S_{+1}}$.
Inequality~\eqref{lem:edgediff}
implies: \begin{align*}        
     (*) &\leq  \sum_{\{i,j\} \in E_{S_{+1}}}
     \left(f(x_i)^2 + f(x_j)^2 - \abs*{f(x_i) - f(x_j)} + \mathds{1}_{|x_i-x_j|=1}\right)\\
      &= \sum_{i \in V} d_i f(x_i)^2 - 
      \sum_{\{i,j\} \in E_{S_{+1}}}
      \left(\abs*{f(x_i) - f(x_j)} 
      - \mathds{1}_{|x_i-x_j|=1}\right).
\end{align*}
Using above in~\eqref{eqn:total_exp}, we get
\begin{align*}
\E{\sum_{i \in V} d_i \cdot f(X_i(t+1))^2 }
\leq   & \left(1-\frac{1}{n}\right)\sum_{i \in V}  d_i \cdot f(x_i)^2\\
\phantom{=}   & + \frac{1}{n} \left(\sum_{i \in V}  d_i \cdot f(x_i)^2 -\sum_{\{i,j\} \in E_{S_{+1}}} \left(\abs*{f(x_i) - f(x_j)}
- \mathds{1}_{|x_i-x_j|=1}\right) 
\right)\\
   =  &  \sum_{i \in V} d_i\cdot f(x_i)^2 - \frac{1}{n}\left(\sum_{\{i,j\} \in E_{S_{+1}}} \; 
   \left(\abs*{f(x_i) - f(x_j)} - 
   \mathds{1}_{|x_i-x_j|=1}\right)\right).
\end{align*}
The above inequality is equivalent to~\eqref{EsumClaim1}.

\subsection{Proof of \autoref{lem:2valpullvoteX}}
\label{sec:ProofVoting}

Our proof of \autoref{lem:2valpullvoteX} is based on the analysis of the convergence time of 2-value voting schemes
given by Cooper and Rivera~\cite[Section 4]{CooperR16}. 
We expand that analysis to obtain a more detailed bound on the convergence time,
which is parameterized by
the initial imbalance of the support for the two opinions. 

Recall that 
for the 2-value pull voting, $S(t)$ is the set of nodes with opinion $1$ at step~$t\ge 0$, and 
$\eta(t) \coloneqq \eta(S(t)) = \min\{\vol{S(t)}, \voll(\ovl S(t))\}$. 
Cooper and Rivera~\cite{CooperR16} study the evolution of 
$\sqrt{\eta(t)}$
in the wider context of 2-value \emph{linear voting} processes, which include the 2-value (asynchronous) 
pull voting.
Tracing the analysis in~\cite[Section 4]{CooperR16}, one can derive the following bound 
for any $S\subseteq V$ such that $\eta(S)>0$
(that is, for any $S\subseteq V$ other than 
$S=\emptyset$ or $S=V$).
\begin{align}\label{LV-onestep-bound}
 \E{\sqrt{\eta(t+1)} \mid S(t) = S} \le \sqrt{\eta(S)} - \frac{\Upsilon}{8\sqrt{\eta(S)}},
\end{align}
where $\Upsilon\coloneqq \Phi(G)\cdot {d_{\min}}/n$.
To recover~\eqref{LV-onestep-bound} from~\cite{CooperR16}, 
look there for Inequality~(11), which bounds $\E{\sqrt{\eta(t+1)} \mid S(t) = S}$
in terms of $\sqrt{\eta(S)}$ and a quantity $\Upsilon(S)$, and for (16) and (17),
which together provide the bound $\Upsilon(S) \ge \Upsilon \cdot \eta(S)$ for the case of
pull voting. 
Note also that~\cite{CooperR16} uses $\mu(U) \coloneqq \vol{U}/(2m)$ 
as the measure of the size of a set of nodes $U\subseteq V$, so in the analysis there
$\eta(S)\coloneqq \min\{\mu(S), \mu(\ovl S)\}$
and $\mu^* \coloneqq \min_{v\in V} \mu(v) = d_{\min}/(2m)$.
A straightforward scaling 
by the factor of $2m$ transforms the quantities and bounds from~\cite{CooperR16} 
to the notation in our paper.

Using~\eqref{LV-onestep-bound}, we now show the following relation between 
the expectations $\E{\sqrt{\eta(t)}}$ and $\E{\sqrt{\eta(t+1)}}$, 
assuming that $\emptyset\neq S(0) \neq V$
(so $\eta(S(0)) >0$):
\begin{align}\label{Eta-one-step-change}
\E{\sqrt{\eta(t+1)}} \leq \E{\sqrt{\eta(t)}} 
    - \frac{\Upsilon}{8}\cdot\pr{\eta(t) > 0}^2 \cdot  \frac{1}{\sqrt{\eta(S(0))}}.
\end{align}
To derive~\eqref{Eta-one-step-change}, we view $\E{\sqrt{\eta(t+1)}}$ as the total expectation 
over the range of the values of $S(t)$.
\begin{align}
    \E{\sqrt{\eta(t+1)}} 
    &= \sum_{S \subseteq V,\: \eta(S)>0}
    \E{\sqrt{\eta(t+1)} \mid S(t) = S} \cdot \pr{S(t) = S}\nonumber\\
    &\le \sum_{S \subseteq V,\: \eta(S)>0}
    \left(\sqrt{\eta(S)} - \frac{\Upsilon}{8\sqrt{\eta(S)}}\right)\cdot\pr{S(t) = S}\nonumber\\
    &= \sum_{S \subseteq V,\: \eta(S)>0}
    \sqrt{\eta(S)} \cdot \pr{S(t) = S}\nonumber\\
    &\;\;\;\; - \; \frac{\Upsilon}{8}\cdot\pr{\eta(t)>0}
    \cdot\sum_{S \subseteq V,\: \eta(S)>0}
    \frac{1}{\sqrt{\eta(S)}}\cdot\pr{S(t) = S \mid \eta(t)>0} \nonumber\\
    &= \E{\sqrt{\eta(t)}} 
    - \frac{\Upsilon}{8}\cdot\pr{\eta(t) > 0}  
    \cdot \E{\frac{1}{\sqrt{\eta(t)}} \mid \eta(t) > 0} \nonumber\\
    &\leq \E{\sqrt{\eta(t)}} 
    - \frac{\Upsilon}{8}\cdot\pr{\eta(t) > 0} \cdot \frac{1}{\E{\sqrt{\eta(t)}\mid \eta(t) > 0}} \label{ByJensen}\\
    &= \E{\sqrt{\eta(t)}} 
    - \frac{\Upsilon}{8}\cdot\pr{\eta(t) > 0}^2 \cdot  \frac{1}{\E{\sqrt{\eta(t)}}}\label{toatlExpect}\\
    &\leq \E{\sqrt{\eta(t)}} 
    - \frac{\Upsilon}{8}\cdot\pr{\eta(t) > 0}^2 \cdot  \frac{1}{\sqrt{\eta(S(0))}}.\label{byMartingale}
\end{align}
Inequality~\eqref{ByJensen} follows by Jensen's Inequality,
Equality~\eqref{toatlExpect} holds because 
\[ \E{\sqrt{\eta(t)}} = \E{\sqrt{\eta(t)}\mid \eta(t)>0} \cdot \pr{\eta(t) > 0},
\]
and Inequality~\eqref{byMartingale} holds because $\left(\sqrt{\eta(t)}\right)_{t\ge 0}$ is a supermartingale, so
$\E{\sqrt{\eta(t)}} \le \sqrt{\eta(0)} = \sqrt{\eta(S(0))}$.
The sequence $\left(\sqrt{\eta(t)}\right)_{t\ge 0}$ is a supermartingale because
both $(\vol{S(t)})_{t \geq 0}$ and $(\voll(\ovl S(t)))_{t \geq 0}$ are martingales
(as shown in~\cite{CooperR16}), 
so their minimum $(\eta(t))_{t \geq 0}$ is a supermartingale, and
as function $\sqrt{\cdot}$ is concave and non-decreasing, then by Jensen's inequality 
$\left(\sqrt{\eta(t)}\right)_{t \geq 0}$ is also a supermartingale.

Let $T = \min\Bigl\{t \ge 0 : \pr{\eta(t)>0} < \tfrac12 \Bigr\}$, implying that 
the probability that the consensus is reached by step $T$ is at least $1/2$.
This time $T$ exists, because
$\lim_{t\rightarrow \infty}\pr{\eta(t)>0} = 0$.
Indeed, there is $\tau = \tau(n)$ such that for each configuration, there is a sequence 
of $\tau$ interactions which leads to consensus. 
This implies that  
$\pr{\eta(t)>0}\le (1-p)^{\lfloor t/\tau\rfloor} \overset{t\to\infty}{\longrightarrow} 0$, 
where $p>0$ is a lower bound on the probability that the next $\tau$ random interactions will be 
exactly a given sequence of interactions.

As $\pr{\eta(t)>0} \ge \tfrac12$ for $0 \le t <T$, 
summing~\eqref{Eta-one-step-change} for $t=0,1,\ldots,T-1$ gives
\begin{align}\label{xmkl23w}
\E{\sqrt{\eta(T)}} \leq \sqrt{\eta(S(0))} 
    - T\frac{\Upsilon}{32}  \frac{1}{\sqrt{\eta(S(0))}}.
\end{align}
Since we must have $\E{\sqrt{\eta(T)}} \ge 0$, then \eqref{xmkl23w}  implies that 
$T = O({\eta(S(0))}/\Upsilon) = O\left(\tfrac{\eta(S(0)) \cdot n}{\Phi(G) \cdot \dmin}\right)$.

\section{Tools}
\subsection{Combinatorial Tools}
The spectral gap characterizes the presence of sparse cuts and the mixing behavior of
random walks, a relationship made precise by Cheeger’s inequality.

\begin{lemma}[Cheeger's inequality {\cite{JerrumSinclair1989,LevinPeresWilmer2009}}]
\label{Cheeger}
Let $G=(V,E)$ be a connected, undirected graph with random-walk matrix $P$. Let $\lambda_2(P)$ 
be the second largest eigenvalue of $P$.
Then
\begin{align*}
\tfrac{\Phi(G)^2}{2}
\;\le\;
1-\lambda_2(P)
\;\le\;
2\Phi(G).
\end{align*}
As $1-\lambda(G)\le 1-\lambda_2(P)$, then 
the upper bound above applies also to 
the absolute spectral gap, 
that is $1-\lambda(G)\le 2\Phi(G)$.
(The lower bound above applies to $1-\lambda_2(P)$, but not generally to the absolute spectral gap.)
\end{lemma}

Beyond expansion, the second eigenvalue controls how evenly edges are distributed between
node sets, formalized by the following lemma.

\begin{lemma}[Expander Mixing Lemma {\cite{Chung1997,HooryLinialWigderson2006}}]\label{Lex:expandermixing}
Let $G=(V,E)$ be a connected, undirected graph with random-walk matrix $P$,
and let $\lambda(G)$ be the second largest absolute
eigenvalue of $P$.
Then for all $S,T\subseteq V$,
\begin{align*}
\bigl|e(S,T)-\tfrac{\vol{S}\vol{T}}{2m}\bigr|
\;\le\;
\lambda(G)\sqrt{\vol{S}\vol{T}},
\end{align*}
where $e(S,T)=\sum_{u\in S}\sum_{v\in T}\mathds{1}_{\{u,v\}\in E}$ counts edges between $S$ and $T$, counting each edge wholly inside $S\cap T$ twice. 
For disjoint sets, $e(S,T)=|E(S,T)|$, with $E(S,T)$ the edge set defined in \Autoref{subsec:preliminaries}.
\end{lemma}

\subsection{Tools from Probability Theory}
We consider in this paper discrete-time stochastic processes with a discrete range of random choices
in each step.
Such a process is defined on a probability space $(\Omega,\mathcal{F},P)$
with $\Omega = B^{\mathbb{N}_0} = \{\omega = (\omega_0, \omega_1, \ldots): \omega_t\in B,\:
\text{for each}\; t\ge 0\}$, where $B$ is the range of random choices
in one step. In DIV on a graph $G=(V,E)$, the set $B = \Vec{E}$ represents all choices for 
the interaction in a given step. 
Let $\mathcal{F}_t=\sigma(\omega_0,\ldots,\omega_{t-1})$, with $\mathcal{F}_0=\{\emptyset,\Omega\}$ for a deterministic initial configuration. Thus $\mathcal{F}_t$ records the interactions in steps $0,\ldots,t-1$, and the pre-update state $\X{t}$ is $\mathcal{F}_t$-measurable. For $t\ge1$, the atomic events are
\[
\{\omega\in\Omega:(\omega_0,\ldots,\omega_{t-1})=(b_0,\ldots,b_{t-1})\},\qquad(b_0,\ldots,b_{t-1})\in B^t.
\]
We further note that for each $t\ge 0$, $\mathcal{F}_t \subseteq \mathcal{F}_{t+1}$
and $\mathcal{F}_t\subseteq\mathcal{F}$ 
(the $\sigma$-algebra $\mathcal{F}$ of all events is generated by $\bigcup_{t\ge 0} \mathcal{F}_t$).

A sequence $(Y_t)_{t\ge 0}$ of random variables defined on $\Omega$
\emph{is adapted to} $(\mathcal{F}_t)_{t\ge 0}$, if for every time $t\ge 0$, the random variable 
$Y_t$ is $\mathcal{F}_t$-measurable, meaning that each atomic event in $\mathcal{F}_t$
determines the value of $Y_t$.
We are interested in the conditional expectation 
$\E{Y_{t+1}\mid \mathcal{F}_t}$, which is
the expectation of $Y_{t+1}$ given the random choices in steps $0,1,\ldots,t-1$. The interaction at step $t$ is the next random choice, which will determine the value of
$Y_{t+1}$.
Formally, $\E{Y_{t+1}\mid \mathcal{F}_t}$ is a mapping $\alpha \mapsto \E{Y_{t+1}\mid \alpha_t=\alpha}$,
where the variable $\alpha_t$ ranges over the set of 
atomic events in $\mathcal{F}_t$.
It often happens that we do not need to have the full knowledge of an atomic event 
$\alpha\in \mathcal{F}_t$ to determine $\E{Y_{t+1}\mid \alpha_t = \alpha}$.
For example, it may suffice to know the 'configuration' $C_t = C_t(\alpha_t)$ 
of the system at step $t$, that is, we may have 
$\E{Y_{t+1}\mid \alpha_t = \alpha} = \E{Y_{t+1}\mid C_t=C_t(\alpha)}$
for each atomic event 
$\alpha\in \mathcal{F}_t$.
In such a case, we may write $\E{Y_{t+1}\mid C_t=C}$, or simply $\E{Y_{t+1}\mid C_t}$,
instead of $\E{Y_{t+1}\mid \mathcal{F}_t}$,
making the underlying filtration implicit.
We use such notation, for example, 
in~\eqref{eq:decreasingpotential}.

A sequence $(Y_t)_{t\ge0}$ is a \emph{martingale} (respectively, \emph{supermartingale}, 
\emph{submartingale}),
if it is adapted to $(\mathcal{F}_t)_{t\ge0}$, satisfies $\E{|Y_t|}<\infty$ for every $t$, and for all $t\ge0$, $\E{Y_{t+1}\mid \mathcal{F}_t}=Y_t$
(respectively, $\E{Y_{t+1}\mid \mathcal{F}_t}\le Y_t$, $\E{Y_{t+1}\mid \mathcal{F}_t}\ge Y_t$).
Analogous definitions apply to a finite sequence $(Y_t)^{T}_{t=0}$ of random variables
adapted to a filtration $(\mathcal{F}_t)^{T}_{t=0}$.
A~random variable $\tau$ is a \emph{stopping time} with respect to $(\mathcal{F}_t)_{t\ge 0}$, 
if for all $t$ the event $\{\tau=t\}$ is measurable with respect to $\mathcal{F}_t$.
In our analysis of DIV, we make use of the following well-known properties of martingales.  
See, for example,~\cite{durrett2019probability} for Doob's Optional Stopping Theorem,
\cite{boucheron2013concentration} for Azuma--Hoeffding inequality
and~\cite{mitzenmacher2017probability} for Chernoff bounds.

\begin{theorem}[Doob's Optional Stopping Theorem]
\label{doob}
Let $(Y_t)_{t\ge0}$ be a supermartingale and let $\tau$ be a stopping time such that
$\tau\le c$ almost surely for some constant $c$.
Then $\E{Y_\tau}\le \E{Y_0}$.
\end{theorem}
\begin{lemma}[Azuma--Hoeffding inequality]\label{AzumaHoeffding}
Let $(X_t)_{t=0}^T$ be a sequence of random variables and
assume that $|X_t-X_{t-1}|\le c$ almost surely for all $1 \le t\le T$.
Then for all $\delta>0$,
\begin{align}
&\pr{X_T \ge X_0 + \delta}
\;\le\; \exp\!\left(-\tfrac{\delta^2}{2Tc^2}\right), \;\; \text{if $(X_t)_{t=0}^T$ is a supermartingale,}\label{AZHofsuperm}\\
&\pr{X_T \le X_0 -\delta}
\;\le\; \exp\!\left(-\tfrac{\delta^2}{2Tc^2}\right), \;\; \text{if $(X_t)_{t=0}^T$ is a submartingale,}\label{AZHofsubm}\\
&\pr{|X_T-X_0|\ge \delta}
\;\le\; 2\exp\!\left(-\tfrac{\delta^2}{2Tc^2}\right), \;\; \text{if $(X_t)_{t=0}^T$ is a martingale.}\label{AZHofm}
\end{align}
\end{lemma}

\begin{lemma}[Chernoff bounds]\label{Chernoff}
Let $X_1, \ldots, X_n$ be independent $0$-$1$ Bernoulli random variables with parameters
$p_1, \ldots, p_n$, respectively, 
$S_n = \sum_{i=1}^{n}X_i$,
$\mu_n = \sum_{i=1}^{n}p_i$. 
(We note that $(S_j - \mu_j)_{j=1}^{n}$ is a martingale.)
If $\mu=\mu_n>0$, then for $0\le\delta\le\mu$,
\begin{align*}
\pr{S_n \ge \mu + \delta} &\le {e^{-\delta^2/(3\mu)}},\\
\pr{S_n \le \mu -\delta} &\le {e^{-\delta^2/(2\mu)}}.
\end{align*}
\end{lemma}

\bibliography{bibliography}

\end{document}